\documentclass[]{iopart}
\usepackage{amsmathred, latexsym, amssymb, amscd,amsfonts, amsthm, epsfig,color,times,appendix}
\usepackage{graphicx}
\usepackage{pictex}
\usepackage[english]{babel}

\newtheorem{example}{Example}[section]

\newtheorem{theorem}{\bf Theorem}[section]

\newtheorem{definition}[theorem]{\textsl{\bf Definition}{}}
\newtheorem{lemma}[theorem]{\bf Lemma}

\begin{document}

\title[Quantum learning: optimal classification of qubit states]
{{\sf \bfseries Quantum learning: optimal classification of qubit states}}
\author{{\sf \bfseries 
M\u{a}d\u{a}lin Gu\c{t}\u{a}$^{1}$ and 
Wojciech  Kot{\l}owski$^{2}$}}

\address{$^{1}$School of Mathematical Sciences, University of Nottingham,
University Park, NG7 2RD Nottingham, United Kingdom \\
$^{2}$ CWI, Science Park 123, 1098 XG Amsterdam\\ 
E-mail: \textcolor[rgb]{0.00,0.00,1.00}{madalin.guta@nottingham.ac.uk}}
\date{\today}

\begin{abstract}
Pattern recognition is a central topic in Learning Theory with numerous applications  such as voice and text recognition, image analysis, computer diagnosis. The statistical set-up in classification is the following: we are given an i.i.d. training set 
$(X_{1},Y_{1}),\dots (X_{n},Y_{n})$ where $X_{i}$ represents a feature and 
$Y_{i}\in \{0,1\}$ is a label attached to that feature. The underlying joint  distribution of $(X,Y)$ is unknown, but we can learn about it from the training set and we aim at devising low error classifiers $f:X\to Y$ used to predict the label of new incoming features. 

Here we solve a quantum analogue of this problem,  namely  the classification of  two arbitrary {\it unknown} qubit states. Given a number of `training' copies from each of the states, we would like to `learn' about them by performing a measurement on the training set. The outcome is then used to design mesurements for the classification of future systems with unknown labels. We find the  asymptotically optimal classification strategy and show that typically, it performs strictly better than a plug-in strategy based on state estimation.

The figure of merit is the {\it excess risk} which is the difference between the probability of error and the probability of error of the optimal measurement when the states are known, that is the Helstrom measurement. We show that the excess risk has rate $n^{-1}$ and compute the exact constant of the rate.

\end{abstract}


\clearpage

\tableofcontents

\title[Quantum learning: optimal classification of qubit states]
\bigskip

\section{Introduction}

Statistical learning theory \cite{Mitchell97,DevroyeGyorfiLugosi96,Vapnik98,FriedmanHastieTibshirani03} is 
a broad research field stretching over statistics and computer science, whose general goal is to devise algorithms which have the ability to learn from data. One of the central learning problems is how to recognise patterns \cite{Bishop}, with practical applications in speech and text recognition, image analysis, computer-aided diagnosis, data mining.

The paradigm of Quantum Information theory is that quantum systems carry a new type of information with potentially revolutionary applications such as faster computation and secure communication \cite{NC00}. Motivated by these theoretical challenges, Quantum Engineering is developing new tools to control and accurately 
measure individual quantum systems \cite{Wiseman&Milburn}. In the process of engineering exotic quantum states, statistical validation has become a standard experimental procedure \cite{Smithey,Haffner} and Quantum Statistical Inference has passed from its purely theoretical status in the 70's \cite{Holevo,Helstrom} to a more practically oriented theory at the interface between the classical and quantum worlds \cite{Leonhardt,Hayashi.editor,Paris.editor,Barndorff-Nielsen&Gill&Jupp}. 

In this paper we put forward a new type of quantum statistical problem inspired by learning theory, namely {\it quantum state classification}.  Similar ideas have already appeared in the physics \cite{SasakiCarlini2002,BergouHillery2005,Hayashi&Horibe,Hayashi_et_al2006} and learning  \cite{AimeurBrassardGambs2006,AimeurBrassardGambs2007,Gambs2008} literature but here we emphasise the close connection with learning and we aim at going beyond the special models based on group symmetry and pure states. However, we limit ourselves to a two dimensional state which could be regarded as a toy model from the viewpoint of learning theory, but hope that more interesting applications will follow. 

Before explaining what quantum classification is, let us briefly mention the classical set-up we aim at generalising. In \emph{supervised learning} the goal is to learn to predict an output $y \in \mathcal{Y}$, given the input (object) $x \in \mathcal{X}$, where input and output are assumed to be correlated and have an {\it unknown} joint distribution $\mathbb{P}$ over $\mathcal{X}\times \mathcal{Y}$. 
To do this, we are first provided with a set of $n$ previously observed inputs 
with known output variables (called \emph{training examples}),
i.e. independent random pairs $(X_i,Y_i), i=1,\ldots,n$ drawn from  $\mathbb{P}$. Using the training set, we construct a function 
$h_{n} \colon \mathcal{X} \to \mathcal{Y}$ to predict the output for future, yet unseen objects. When $\mathcal{Y} = \{0,1\}$, i.e. the output is a binary variable, this is called \emph{binary classification} and is the typical set-up in pattern recognition. The input space is usually considered to be a subset of $p$-dimensional space $\mathbb{R}^p$, so that the object $x$ can be described by $p$ measurement values often called \emph{features}. This description is very general as it allows e.g. to handle categorical (non-numerical) values (encoded as integer numbers), images (e.g. measured brightness of each pixel corresponds to a separate feature), time series (features corresponds to the values of the signal at given times), etc. 

In this paper, we consider the classification problem in which the objects to be classified are quantum states. Simply, we have a quantum system prepared in either of two {\it unknown} quantum states and we want to know which one it is. 
As in the classical case, this only makes sense if we are also provided with training examples from both states, with their respective labels, from which we can learn about the two alternatives. How could such a scenario occur? Suppose we send one bit of information through a noisy quantum channel which 
is not known. To decode the information (the input in this case) we need to be able to classify the output states corresponding to the two inputs. Alternatively, the binary variable may be related to a coupling of the channel which we want to detect. 

Needless to say, quantum systems are {\it intrinsically statistical} and can be `learned' only by repeated preparation, so that the problem is really the quantum extension of the classical classification problem. On the other hand this is related to the problem of state discrimination which in the case of two hypotheses, has an explicit solution known as the Helstrom measurement \cite{Helstrom}. The point is that when the states are unknown, the Helstrom measurements is itself unknown and has to be learned from the training set. An intuitive solution would be a \emph{plug-in} procedure: first estimate the two states, and then apply the Helstrom measurement corresponding to the estimates on any new to-be-classified state. This indeed gives a reasonable classification strategy, but as we will see, 
this is {\it not} the best one. The optimal strategy in the asymptotic framework is to directly estimate the Helstrom measurement without intermediate states estimation. The optimality is defined by the natural figure of merit called {\it excess risk}, which is the difference between the expected error probability and the error probability of the Helstrom measurement. We show that the excess risk converges to zero with the size of the training set as $n^{-1}$ and the ratio 
between the optimal and state estimation plug-in risk is a constant factor.

Our analysis is valid for arbitrary mixed states and is performed in a {\it pointwise, local minimax} (rather than Bayesian) setting which captures the behaviour of the risk around any pair of states. The key theoretical tool is the 
recently developed theory of local asymptotic normality (LAN) for quantum states \cite{Guta&Kahn,Guta&Janssens&Kahn,Guta&Kahn2,Guta&Jencova} which is an extension of the classical concept in mathematical statistics introduced by Le Cam \cite{LeCam}. Roughly, LAN says that the collective state $\rho_{\theta}^{\otimes n}$ of $n$ i.i.d. quantum systems can be approximated by a simple Gaussian state of some classical variables and quantum oscillators. 
This was used to derive optimal state estimation strategies for arbitrary mixed states of arbitrary finite dimension, and also in finding quantum teleportation benchmarks for multiple qubit states \cite{Guta&Adesso}. In this paper, LAN is used to identify the (asymptotically) optimal measurement on the training set as 
{\it linear measurement} on two harmonic oscillators. Similarly to the case of state estimation such collective measurements perform strictly better than the local 
ones \cite{Bagan_et_al2004,Bagan_et_al2006}. Moreover, optimal learning collective measurement is different from the optimal measurement for state estimation, showing once again that generically, different quantum decision problems cannot be solved optimally simultaneously. 

{\it Related work.} Sasaki and Carlini \cite{SasakiCarlini2002} defined a  \emph{quantum matching machine} which aims at pairing a given `feature' state with the closest out of a set of `template' states. The problem is formulated in a Bayesian framework with uniform priors over the feature and template pure states which are considered to be unknown. 
Bergou and Hillery \cite{BergouHillery2005} introduced a discrimination machine, which corresponds to our set-up in the special case when the training
set is of size $n=1$. The papers \cite{Hayashi&Horibe,Hayashi_et_al2006} deal with the problem of quantum state identification as defined in this paper. The special case of Bayesian risk with uniform priors over pure states was solved in \cite{Hayashi&Horibe}, with the small difference that the learning and classification steps are done in a single measurement over $n+1$ systems. However, as in the case of state estimation \cite{Bagan&Gill}, the proof relies on the special symmetry of the prior and does not cover mixed states. Finally, the concept of quantum classification  was already proposed in a series of papers \cite{AimeurBrassardGambs2006,AimeurBrassardGambs2007,Gambs2008}. 
However, the authors mostly focused on problem formulation, reduction between different problem classes and general issues regarding learnability. Other related papers which fall outside the scope of our investigation are \cite{Gammelmark&Molmer,Bisio&Chiribella}.

This paper is organised as follows. Section \ref{sec.2} gives a short overview of the classical classification set-up and introduces its quantum analogue. 
Section \ref{sec.lan}  discusses the LAN theory with emphasis on the qubit case.  
In section \ref{sec.local.formulation} we reformulate the classification problem in the asymptotic (local) framework, as an estimation problem with quadratic loss for the training set. The main result is Theorem \ref{th.main} of Section \ref{sec.theorem} which gives the mimimax excess risk for the case of known priors. The case of unknown priors is treated Section \ref{sec.priors}. The optimal classifier is compared to the plug-in procedure based on optimal state estimation in Section \ref{sec.comparison.plugin}. 
The geometry of the problem is captured by the  Bloch ball illustrated in Figure 
\ref{fig.bloch2}.  We conclude the paper with discussions.  
\section{Classical and quantum learning}
\label{sec.2}
\subsection{Classical Learning}

Let $(X,Y)$ be a pair of random variables with joint distribution $\mathbb{P}$ over the measure space $(\mathcal{X} \times \{0,1\},\Sigma)$. In the classical setting $\mathcal{X}$ is usually a subset of $\mathbb{R}^p$ 
and $Y$ is a binary variable.

In a first stage we are given a {\it training} set of $n$ i.i.d. pairs $\{(X_1,Y_1),\ldots,(X_n,Y_n)\}$ with distribution $\mathbb{P}$, from which we would like to `learn' about $\mathbb{P}$. In the second stage we are presented with a new sample $X$ and we are asked to guess its unseen label $Y$. For this we construct a (random) {\it classifier} 
$$
\hat{h}_{n} \colon \mathcal{X} \to \{0,1\}
$$ 
which depends on the data $ (X_1,Y_1),\ldots,(X_n,Y_n)$. Its overall accuracy is measured in terms of the \emph{expected error rate} 
according the data distribution $\mathbb{P}$, 
$$
\mathbb{P}_e(\hat{h}_{n}) = \mathbb{P}(\hat{h}_{n}(X) \neq Y) = \mathbb{E}[1_{\hat{h}_{n}(X) \neq Y}],
$$
where $1_C$ is the indicator function equal to $1$ if $C$ is true, and $0$ otherwise.
However the error rate itself does not give a good indication on the performance of the learning method. Indeed, even an `oracle' who knows $\mathbb{P}$ exactly has typically a non-zero error: in this case the optimal $\hat{h}$ is the 
{\it Bayes classifier}  which chooses the label that is more probable with respect to conditional distribution 
$\mathbb{P}(y|x)$
\begin{equation}\label{eq.Bayes.classifier}
h^{*} (X) = 
\left\{
\begin{array}{ccc}
0 & if &\eta(X) \leq 1/2\\
1& if  &\eta(X) > 1/2
\end{array}
\right.
\end{equation}
where $\eta(x):= \mathbb{P}(Y=1 |x)$. The {\it Bayes risk} is  
$$
\mathbb{P}_e(h^{*}) = 
\mathbb{E}[\mathbb{E}[1_{\hat{h}^{*}(X) \neq Y}|X]]=  \frac{1}{2} \left(1 - \mathbb{E}[|1 - 2\eta(X)|]\right).
$$
An alternative view of the Bayes classifier which fits more naturally in the quantum set-up is the following. We are given data $X$ whose probability distribution is 
either $\mathbb{P}_{0}(X):= \mathbb{P}(X|Y=0)$ or 
$\mathbb{P}_{1}(X):= \mathbb{P}(X|Y=1)$ and we would like to test between the two hypotheses. We are in a Bayesian set-up where the hypotheses are chosen randomly with prior distributions $\pi_{i}= \mathbb{P}(Y=i)$. The optimal solution of this problem is the well known likelihood ratio test: we choose the hypothesis with higher likelihood
$$
h^{*} (X) = 
\left\{
\begin{array}{ccc}
0 & if &  \pi_{0} \mathbb{P}_{0}(X)>\pi_{1} \mathbb{P}_{1}(X) \\
&\\
1& if  & \pi_{0} \mathbb{P}_{0}(X)\leq\pi_{1} \mathbb{P}_{1}(X)
\end{array}
\right.
$$
which can be easily verified to be identical to the previously defined Bayes classifier. The Bayes risk can be written as 
\begin{equation}\label{bayes.risk.classical}
\mathbb{P}_{e}^{*}= \frac{1}{2}(1- \| \pi_{0}p_{0} -\pi_{1} p_{1}\|_{1} ),
\end{equation}
where $p_{i}$ are the densities of $\mathbb{P}(X|Y=i)$ with respect to some common reference measure.

Returning to the classfication set-up where $\mathbb{P}$ is unknown,  we  see 
that a more informative performance measure for $\hat{h}_{n}$ is the \emph{excess risk}:
\begin{equation}
\label{eqn:excess_risk}
R(\hat{h}_{n}) = \mathbb{P}_e(\hat{h}_{n}) - \mathbb{P}_e(h^{*})\geq 0
\end{equation}
which measures how much worse the procedure $\hat{h}_{n}$ performs compared to the performance of the oracle classifier. In statistical learning theory one is primarily interested in consistent classifiers, for which the excess risk converges to $0$ as $n \to \infty$, and then in finding classifiers with fast convergence rates \cite{DevroyeGyorfiLugosi96,Vapnik98}. But how to compare different learning procedures? One can always design algorithms which work well for certain distributions and badly for others. Here we take the statistical approach and consider that all prior information about the data is encoded in the {\it statistical model} $\{\mathbb{P}_{\theta} :\theta\in \Theta\}$ i.e. the data comes from a distribution which depends on some unknown parameter  
$\theta$ belonging to a parameter space $\Theta$. The later may be a subset of 
$\mathbb{R}^{k}$ (parametric) or a large class of distributions with certain `smoothness' properties (non-parametric). One can then define the maximum risk of $\hat{h}_{n}$
$$
R_{max}(\hat{h}_{n}) := \sup_{\theta\in \Theta} R_{\theta}(\hat{h}_{n})
$$
where $ R_{\theta}$ denotes the excess risk when the underlying distribution is $\mathbb{P}_{\theta}$. A procedure $\tilde{h}_{n}$ is called minimax if its 
maximum risk is smaller than that of any other procedure
\begin{equation}\label{eq.max.risk}
R_{max}(\tilde{h}_{n}) = \inf_{\hat{h}_{n}} R_{max}(\hat{h}_{n})=\inf_{\hat{h}_{n}}\sup_{\theta\in \Theta} R_{\theta}(\hat{h}_{n}) .
\end{equation}
Alternatively one can take a Bayesian approach and optimise the average risk with respect to a given prior over $\Theta$.

\begin{example}\label{ex.coin.toss} 
Let $(X,Y)\in\{0,1\}^{2}$ with unknown parameters 
$\eta(0),\eta(1)$ and $\mathbb{P}(X=0)$, satisfying $\eta(0) <1/2$ and 
$\eta(1)>1/2$. Then the Bayes classifier is $h^{*}(0)=0$ and $h^{*}(1)=1$. On the other hand, from the training sample one can estimate $\eta(i)$ and obtain the concentration result
$$
\mathbb{P}[\hat{\eta}_{n}(0)<1/2 ~{\rm and}~ \hat{\eta}_{n}(1)>1/2 ] =1- O(\exp(-cn)).
$$
Thus the plug-in estimator $\hat{h}_{n}$ obtained by replacing $\eta$ by $\hat{\eta}_{n}$ in \eqref{eq.Bayes.classifier} is equal to $h^{*}$ with high probability and 
the excess risk is exponentially small.
\end{example}

The crucial feature leading to exponentially small risk was the fact that the regression function $\eta(X)$ is bounded away from the critical value $1/2$. This situation is rather special but shows that the behaviour of the excess risk depends on the properties of $\eta$ around the value $1/2$. Let us look at another simple example with a different behaviour.
\begin{example}\label{ex.gaussian.learning}
Let $(X,Y)\in \mathbb{R}\times \{0,1\}$ with 
$$\mathbb{P}(X|Y=0)= N(a, 1), \qquad \mathbb{P}(X|Y=1)= N(b, 1)
$$ 
for some unknown means $a<b$, and $\mathbb{P}(Y=0)=1/2$. From Figure \ref{fig.gaussians} we can see that $p_{0} (x) \leq p_{1}(x)$ if and only if 
$x\geq (a+b)/2$ so that the Bayes classifier is 
$$
h^{*}(x) =\left\{
\begin{array}{ccc}
0 & if & x <(a+b)/2\\
1& if  & x \geq (a+b)/2
\end{array}
\right.
$$
The Bayes risk is equal to the orange area under the two curves. Again a natural classifier is obtained by estimating the midpoint $(a+b)/2$ and plugging into the above formula. The additional error is the area of the green triangle. 
Since $(\hat{a}+\hat{b})/2-(a+b)/2\approx 1/\sqrt{n}$ one can deduce that 
$$
R(\hat{h}_{n}) =O(n^{-1}),
$$
and it can be shown that this rate of convergence is optimal 
\cite{AudibertTsybakov07}.
\end{example}

\begin{figure}
\begin{center}
\label{fig.gaussians}
\includegraphics[width=6cm]{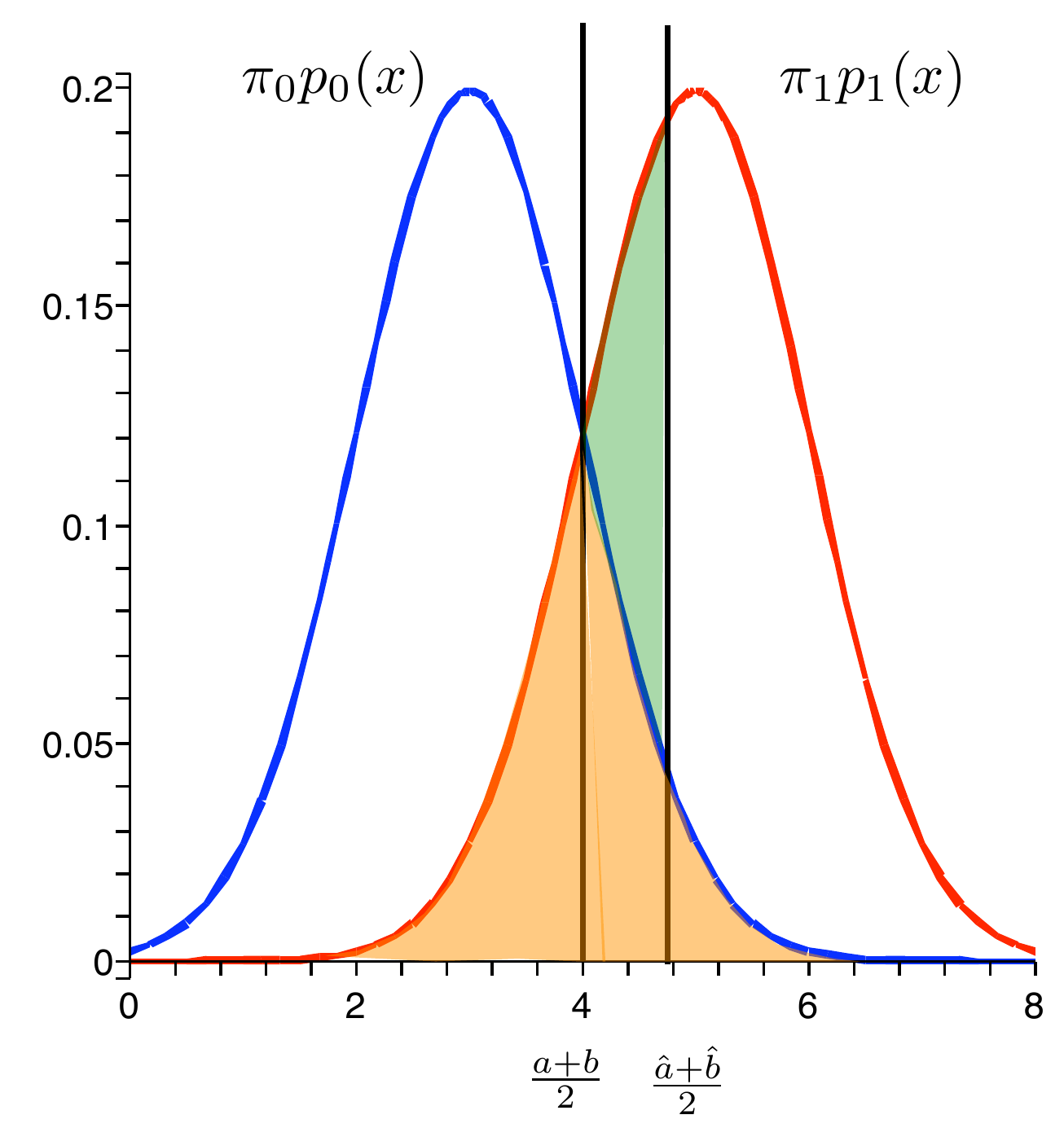}
\caption{Likelihood functions for two normal distributions with means $a,b$. The Bayes risk is the area of the orange triangle. The excess risk is the area of the 
green triangle}
\end{center}
\end{figure}
From this example we see that the rate is determined by the behaviour of the regression function $\eta$ around $1/2$, namely in this case
$$
\mathbb{P}(|\eta(x)-1/2| \leq t) =O (t) , \quad t\geq 0
$$
which is called the {\it margin condition}. Roughly speaking, in a parametric model satisfying the margin condition, the excess risk goes to zero as 
$O\left(n^{-1}\right)$. In non-parametric models (which are the main focus of learning theory), arbitrarily slow rates are possible depending on the complexity of the model and the behaviour of the regression function \cite{AudibertTsybakov07}.

According to Vapnik \cite{Vapnik98}, one of the principles of statistical learning is:
``when solving a problem of interest, do not solve a more general problem as an intermediate step.'' This is interpreted as saying that learning procedures which estimate first the statistical model (or regression function) and then plug this estimate into the Bayes classifier, are less efficient than methods which aim at constructing $\hat{h}(x)$ directly. Recently it has been shown \cite{AudibertTsybakov07} that this is not necessarily the case if some type of margin condition is assumed, and that  plug-in estimators
\begin{equation}
 \label{eqn:plug-in}
\hat{h}_{\textsc{plug-in}}(x) = 1_{\hat\eta(x) \geq 1/2}.
\end{equation}
can perform close to, or at `fast $n^{-1}$ rates'. In this paper we show that at least in what concerns the {\it constant} in front of the rate, 
{\it direct quantum learning performs better than plug in methods based on optimal  state estimation}. This is a purely quantum phenomenon which stems from 
the incompatibility between the optimal measurements for estimation and learning.

\subsection{Quantum Learning}

We now consider the quantum counterpart of the learning problem, the classification of quantum states. In this case, $\mathcal{X}$ is replaced by a Hilbert space of dimension $d$. To find the counterpart of $\mathbb{P}$ we write  
$\mathbb{P}(dx,y)=\mathbb{P}(dx|y)\mathbb{P}(y)$ and replace the conditional distributions $\mathbb{P}(dx|y=0)$ 
and $\mathbb{P}(dx|y=1)$ by density matrices $\rho$ and $\sigma$, 
while $\mathbb{P}(y)$ describes prior probabilities over the states, usually denoted
by $\pi_y := \mathbb{P}(Y=y)$. There is no direct counterpart of the object $x$, since the quantum state is \emph{identified} with its description in terms of 
a density matrix; however, one can think of $x$ as a set of values obtained by
measuring the state $\rho$.

\begin{table}
\caption{Comparison of classical and quantum learning.}
\begin{center}
\label{fig:condition}
\begin{small}
\begin{tabular}{l@{}cc}
\hline
element & classical learning & quantum learning \\
\hline \\[-3mm] 
distribution & $\mathbb{P}$ & $(\rho,\sigma)$ with priors $(\pi_0,\pi_1)$ \\
training example & $(x,y)$ & $(\rho,0)$ or $(\sigma,1)$ \\
training set & $\{(x_1,y_1),\ldots,(x_n,y_n)\}$ & $\rho^{\otimes n_0}\otimes \sigma^{\otimes n_1}$ \\
function & classifier $\hat{h}$ & measurement $\widehat{P}$ \\
optimal function & $h^*(x) = 1_{\eta(x) \geq 1/2}$ & $P^* = [\pi_{0}\rho-\pi_{1}\sigma_1]_{+}$ \\
minimum risk & $\frac{1}{2} \left(1 - \mathbb{E}[|1 - 2\eta(X)|]\right)$ &
 $\frac{1}{2}\left(1 - {\rm Tr}[|\pi_1 \sigma - \pi_0 \rho|]  \right)$ \\
risk & $\mathbb{P}(\hat{h}(X) \neq Y) - \mathbb{P}^*_e$ & 
$\mathbb{E}{\rm Tr}\left[(\pi_1 \sigma - \pi_0 \rho) (\widehat{P} - P^*) \right]$ \\
\hline
\end{tabular}
\end{small}
\end{center}
\end{table}

The training set consists of $n$ i.i.d. pairs
$\{(\tau_1,Y_1),\ldots, (\tau_n,Y_n)\}$, where $\tau_i = \rho$ if $Y_i=0$ and
$\tau_i = \sigma$ if $Y_i = 1$. Thus we are randomly given copies of $\rho$ and 
$\sigma$ together with their labels, but {\it we  do not know} what $\rho$ and 
$\sigma$ are. After a permutation the joint state of the training set can be concisely written as $\rho^{\otimes n_0} \otimes \sigma^{\otimes n_1}$, where $n_y$ is the number of copies for which $Y_{j}=y$.

The experimenter is allowed to make any physical operations on the training set (such as unitary evolution or measurements) and outputs a binary-valued measurement $\mathbb{C}^{2}$  with POVM elements $\widehat{M}_{n}:=(\widehat{P}_{n},\mathbf{1}-\widehat{P}_{n})$. This (random) POVM plays the role of the classical classifier $\hat{h}_{n}$: given a new copy of the quantum state whose label is unknown, we apply the measurement $\widehat{M}_{n}$ to guess whether the state is 
$\rho$ or $\sigma$. The accuracy is measured in terms of the expected misclassification error:
\begin{equation*}
\mathbb{P}_e(\widehat{M}_{n}) =\mathbb{E}\left[ \pi_0 {\rm Tr}[\rho (\mathbf{1} - \widehat{P}_{n})] + \pi_1 {\rm Tr}[\sigma \widehat{P}_{n}] \right] 
\end{equation*}
where the expectation is taken over the outcomes $\widehat{P}_{n}$. 

The Bayes classifier $M^{*}$ is nothing but the Helstrom measurement 
\cite{Helstrom} which optimally discriminates between {\it known} states $\rho,\sigma$ with priors $\pi_{0},\pi_{1}$. In this case $M^{*}= (P^{*},\mathbf{1}-P^{*})$ where $P^*$ is the projection onto the subspace of positive
eigenvalues of the operator $\pi_0 \rho - \pi_1 \sigma$, i.e. $P^* = [\pi_{0}\rho-\pi_{1}\sigma]_{+}$. Note that if both eigenvalues are of the same sign, the
optimal procedure is to choose the state with higher $\pi_i$ without making any measurement at all. The \emph{Helstrom risk} can be expressed as: 
$\mathbb{P}^*_e = \frac{1}{2}\left(1 - {\rm Tr}[|\pi_1 \sigma - \pi_0 \rho|]  \right).$ which is the quantum extension of \eqref{bayes.risk.classical}.

As before, the performance of an arbitrary classifier $\widehat{M}_{n}$ is measured 
by the excess risk:
\begin{equation}
R(\widehat{M}_{n}) = \mathbb{P}_e(\widehat{M}_{n}) - \mathbb{P}^*_e  
= \mathbb{E}{\rm Tr}\left[(\pi_1 \sigma - \pi_0 \rho) (\widehat{P}_{n} - P^*) \right],
\label{eq.quantum_excess}
\end{equation}
which is expected to vanish asymptotically with $n$.

In Table \ref{fig:condition} we summarise the analogous concepts in the classical and the quantum learning set-up. Besides these obvious correspondences we would like to point out some interesting differences. Based on the coin toss example \ref{ex.coin.toss} one may expect that the classification of two qubit states should exhibit similar exponentially fast rates. In fact as we will show in this paper, 
the rate is $n^{-1}$ as in example \ref{ex.gaussian.learning} where the data is not discrete but continuous and the regression function is not bounded away from 
$1/2$. 
A possible explanation is the fact that in the quantum case the `data' to be labelled is a quantum system and the distribution of the outcome depends on the measurement. A helpful way to think about it is illustrated in Figure \ref{fig.controlled.learning}. The unknown label is the input of a black box which outputs the data $X$ with conditional distribution 
$\mathbb{P}(X|Y)$. In the quantum case the box has an additional input, the measurement choice which appears as a parameter in the conditional distribution and is controlled by the experimenter. The game is to learn from the training set the optimal value of this parameter, for which the identification of the label $Y$ is most facile. This set-up resembles that of {\it active learning} \cite{Cohn} where the training data $X_{i}$ are actively chosen rather than collected randomly.

\begin{figure}[ht!]
\begin{center}
\label{fig.controlled.learning}
\includegraphics[width=7cm]{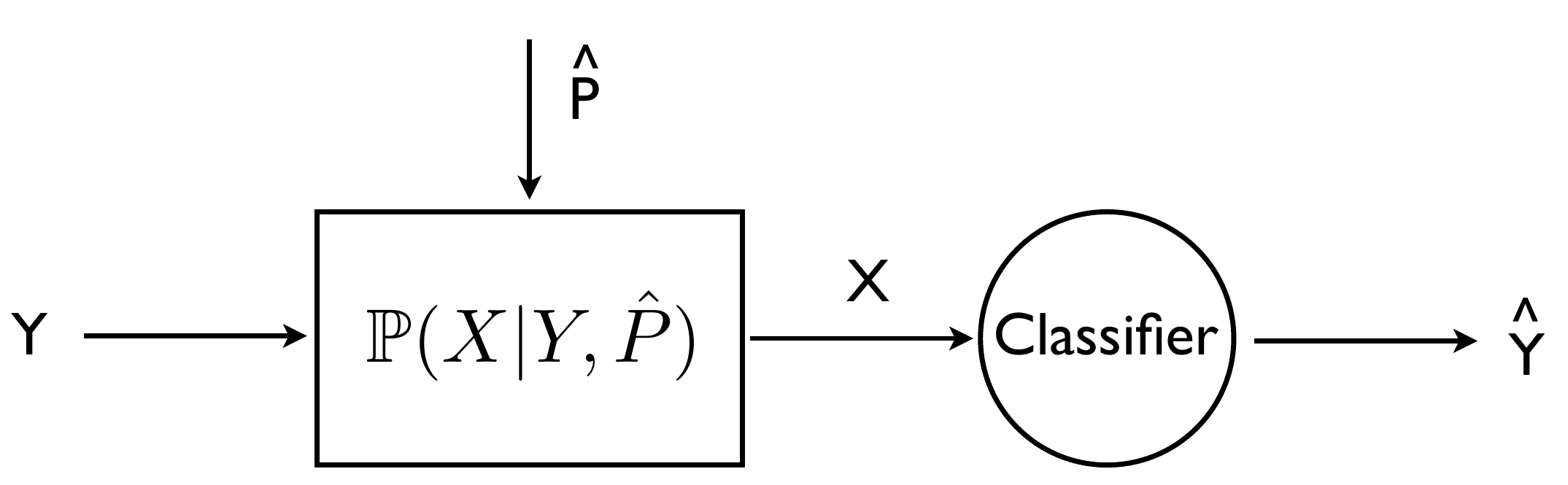}
\end{center}
\caption{Quantum learning seen as classical learning with data distribution depending on on additional parameter controlled by the experimenter}
\end{figure}

\subsection{Local minimax formulation of optimality}
\label{sec.local.minimax.}
We now give the precise formulation of what we mean by asymptotic optimality of a learning strategy $\{\widehat{M}_{n}: n\in \mathbb{N}\}$. As in the classical case we construct a model which contains all unknown parameters of the problem: the two states $\rho, \sigma$ and the prior $\pi_{0}$. We denote these parameters collectively by $\theta$ which belongs to a parameter space 
$\Theta\subset\mathbb{R}^{k}$. When some prior information is available about the model, it can be included by restricting to a sub-model of the general one.
As in the classical case we denote by $R_{\theta}(\widehat{M}_{n})$, the risk of $\widehat{M}_{n}$ at $\theta$, and we can define the maximum risk as in \eqref{eq.max.risk}.
However, assuming for the moment that that the optimal rate of classification is $n^{-1}$, we use a more refined performance measure which is the local version of the maximum risk $R_{max}$ around a fixed parameter $\theta_{0}$
\begin{equation}\label{eq.local.maximum.risk}
R^{(l)}_{max}(\widehat{M}_{n}~;~\theta_{0}):= 
\sup_{\|\theta-\theta_{0}\|\leq n^{-1/2+\epsilon}} n R_{\theta}(\widehat{M}_{n})
\end{equation}
where $\epsilon>0$ is a small number. Note that in the above definition the usual risk was multiplied by the inverse of its rate $n$ so that we can expect 
$R^{(l)}_{max}$ to have a non-trivial limit when $n\to \infty$. The reason for choosing the local maximum risk is that it reflects better the difficulty of the problem in different regions of the parameter space while the maximum risk captures the worst possible behavior over the whole parameter space. 
We can think of the local ball $\|\theta-\theta_{0}\|\leq n^{-1/2+\epsilon}$ as the intrinsic parameter space when the training set consists of $n$ samples. Indeed a simple estimator $\theta_{0}$ on a small proportion $\tilde{n}=n^{1-\epsilon}$ of the sample locates the true parameter in such a ball with high probability (see Lemma 2.1 in \cite{Guta&Janssens&Kahn}).

\begin{definition}\label{def.local.minmax}
The local minimax risk at $\theta_{0}$ is defined as
$$
R^{(l)}_{minmax}(\theta_{0}):= 
\limsup_{n\to\infty} \inf_{\widehat{M}_{n}} R^{(l)}_{max}(\widehat{M}_{n}~;~\theta_{0}).
$$
A sequence of classifiers $\{\tilde{M}_{n}: n\in \mathbb{N}\}$ is called locally asymptotic minimax if 
$$
\limsup_{n\to\infty} R^{(l)}_{max}(\tilde{M}_{n}~;~\theta_{0}) = 
R^{(l)}_{minmax}(\theta_{0}).
$$
\end{definition}

We identify two general learning strategies. The first one consists in estimating the states $\rho,\sigma$ and prior $\pi_0$ (optimally) to get $\hat{\rho},\hat{\sigma},\hat\pi_0$ and then constructing the classifier (measurement)
as:
\begin{equation}\label{eq.plug.in.quantum}
\widehat{P}_{\textsc{plug-in}} = [\hat\pi_0\hat{\rho}-\hat\pi_1\hat{\sigma}]_{+}.
\end{equation}
The second strategy aims at estimating the Helstrom projection $P^{*}$ directly from the training set without passing through state estimation. As we will see, it turns out that in general the latter performs better than the former.

In section \ref{sec.lan} we review the concept of local asymptotic normality which means that locally, the training set can be efficiently approximated by a simple Gaussian model consisting of displaced thermal equlibrium states and classical Gaussian random variables. In section \ref{sec.local.formulation} we show how to reduce the local classification risk for qubits to an expectation of a quadratic form in the local parameters. This will simplify the problem of finding the optimal measurement of the training set, to that of finding the optimal measurement of a Gaussian state for a quadratic loss function \cite{Holevo}.

\section{Local asymptotic normality}\label{sec.lan}

In a series of papers 
\cite{Guta&Kahn,Guta&Janssens&Kahn,Guta&Kahn2} Gu\c{t}\u{a} and Kahn and 
Gu\c{t}\u{a} and Jencova \cite{Guta&Jencova} developed a new approach to state estimation based on the extension of the classical statistical concept of local asymptotic normality \cite{LeCam}. Using this tool one can cast the problem of (asymptotically) optimal state estimation into a much simpler one of estimating the mean of a Gaussian state with known variance. 

Local asymptotic normality provides a convenient description of quantum statistical models involving i.i.d. quantum states which can also be applied to the present learning problem. In this section we will give a brief introduction to this subject in as much as it is necessary for this paper and we refer to \cite{Guta&Kahn2} for proofs and a more in depth analysis.

\subsection{Local asymptotic normality in classical statistics}

A typical statistical problem is the estimation of some unknown parameter 
$\theta$ from a sample $X_{1}, \dots , X_{n}\in \mathcal{X}$ of independent, identically distributed random variables drawn from a distribution $\mathbb{P}_{\theta}$ over a measure space $(\mathcal{X},\Sigma)$. 
If $\theta$ belongs to an open subset of 
$\mathbb{R}^{k}$ for some finite dimension $k$ and if the map 
$\theta\to \mathbb{P}_{\theta}$ is sufficiently smooth, then widely used estimators 
$\hat{\theta}_{n}(X_{1},\dots, X_{n})$ such as the maximum likelihood  are 
asymptotically optimal in the sense that they converge to $\theta$ at a rate 
$n^{-1/2}$ and the error has an asymptotically normal distribution 
\begin{equation}\label{eq.normality}
\sqrt{n}(\hat{\theta}_{n}-\theta)\overset{\mathcal{L}}{\longrightarrow} 
N(0, I^{-1}(\theta)),
\end{equation}
where the right side is the lower bound set by the Cram\'{e}r-Rao inequality for unbiased estimators. To give a simple example, if $X_{i}\in \{0,1\}$ is the result of a coin toss with $\mathbb{P}[X_{i}=1]=\theta$ and $\mathbb{P}[X_{i}=0]=1-\theta$ 
then the sufficient statistic
$$
\hat{\theta}_{n} = \frac{1}{n} \sum_{i=1}^{n} X_{i} 
$$
satisfies \eqref{eq.normality} by the Central Limit  Theorem (CLT). 

Naturally, the first inquiries into quantum statistics concentrated on 
generalising the Cram\'{e}r-Rao inequality to unbiased measurements, and on 
finding asymptotically optimal estimators which achieve the 
quantum version of the Fisher information matrix \cite{Helstrom,Holevo,Belavkin}. However it was found that due 
to the additional uncertainty introduced by the non-commutative nature of quantum mechanics the situation is essentially different from the classical case. 
A summary of these finding is
\begin{enumerate}
\item
the multi-dimensional version of the Cram\'{e}r-Rao bound is in general not achievable;
\item 
the optimal measurement depends on the loss function, i.e. the quadratic form 
$(\hat{\theta}-\theta)^{t}G(\hat{\theta}-\theta)$ and different weight matrices $G$ lead in general to incompatible measurements.
\end{enumerate}

As we will see, these issues can be overcome by adopting a more modern perspective to asymptotic statistics provided by the technique of local asymptotic normality \cite{LeCam,vanderVaart}. Instead of analysing particular estimation problems, the idea is to consider the structure of the statistical model underlying the data and to approximate it by a simpler model for which the statistical problems are 
easy to solve. In order to obtain a non-trivial limit model 
it makes sense to rescale the parameters according to their uncertainty, so we assume that $\theta$ is localised in a region of size $n^{-1/2}$ and we can write 
$\theta=\theta_{0}+ h/\sqrt{n}$ with $\theta_{0}$ known and $h\in  \mathbb{R}^{k}$ the local parameter to be estimated. Such an assumption does not restrict the generality of the problem since one can use an adaptive two-steps procedure where a rough estimate $\theta_{0}$ is obtained in the first step using a small part of the sample, and the rest is used for the accurate estimation of the local parameter $h$.

%



%

Local asymptotic normality means that  the sequence of (local) statistical models 
\begin{equation}\label{eq.seq.local.model}
\mathcal{P}_{n}:=\left\{ \mathbb{P}^{n}_{\theta_{0}+h/\sqrt{n}} : \|h\|< C\right\}, \qquad n\in \mathbb{N}
\end{equation}
depending `smoothly' on $h$, converges to the {\it Gaussian shift model} 
\begin{equation}\label{eq.Gaussian.shift}
\mathcal{G}:=\left\{ N( h, I^{-1}(\theta_{0}) ): \|h\|< C\right\}
\end{equation}
where we observe a single Gaussian variable with mean $h$ and  fixed and known variance. The convergence has a precise mathematical definition in terms of the 
Le Cam distance between two statistical models which quantifies the extent to which each model can be `simulated' by randomising data from the other. 
\begin{definition}
A positive linear map 
$$
T:L^{1}(\mathcal{X},\mathcal{A},\mathbb{P}) \to 
L^{1}(\mathcal{Y},\mathcal{B},\mathbb{Q})  
$$
is called a stochastic operator (or randomisation) 
if $\|T(p)\|_{1}= \|p\|_{1}$ for every $p\in L^{1}_{+}(\mathcal{X})$.
\end{definition}
 For simplicity we consider only dominated models for which all distributions have densities with respect to some fixed reference distribution. In this case a randomisation is the classical analogue of a quantum channel.
\begin{definition}\label{def.LeCam.distance}
Let $\mathcal{P} := \{ \mathbb{P}_{\theta}: \theta\in \Theta\}$ and 
$\mathcal{Q}:= \{\mathbb{Q}_{\theta}: \theta\in \Theta\} $ be two dominated statistical models with distributions having probability densities 
$p_{\theta}:=d\mathbb{P}_{\theta}/d\mathbb{P}$ and 
$q_{\theta}:=d\mathbb{Q}_{\theta}/d\mathbb{Q}$.
The deficiencies 
$\delta(\mathcal{P},\mathcal{Q}) $ 
and 
$\delta(\mathcal{Q},\mathcal{P})$ 
are defined as
\begin{align*}
\delta(\mathcal{P},\mathcal{Q}) &:= 
\inf_{T} \sup_{\theta\in \Theta}\| T(p_{\theta}) -q_{\theta} \|_{1}\\
\delta(\mathcal{Q},\mathcal{P}) &:=  
\inf_{S} \sup_{\theta\in \Theta}
\| S(q_{\theta}) -p_{\theta}\|_{1} 
\end{align*}
where the infimum is taken over all randomisations $T,S$. The Le Cam distance between $\mathcal{P}$ and $\mathcal{Q}$ is 
$$
\Delta(\mathcal{P},\mathcal{Q}):= 
{\rm max}(\delta(\mathcal{Q},\mathcal{P}) ,\, \delta(\mathcal{P} ,\mathcal{Q})).
$$
\end{definition}

With this definitions the local asymptotic normality for i.i.d. parametric models can 
be formulated as 

\begin{theorem}
The sequence of local models \eqref{eq.seq.local.model} converges in the Le Cam distance to the Gaussian shift model \eqref{eq.Gaussian.shift}
$$
\lim_{n\to \infty}\Delta(\mathcal{P}_{n},\mathcal{G})=0.
$$
\end{theorem}
This statement can be extended to slowly increasing local neighbourhoods 
$\|h\|\leq n^{\epsilon}$ with precise convergence rate for the 
Le Cam distance.  

\subsection{Local asymptotic normality in quantum statistics}
\label{sec.qlan}

We will now describe the quantum version of local asymptotic normality for the simplest case of a family of spin states. The general result valid for arbitrary finite dimensional systems can be found in \cite{Guta&Kahn2}.

We are given $n$ spins independent identically prepared in the state 
$$
\rho_{\vec{r}} = \frac{1}{2}(\mathbf{1} + \vec{r}\vec{\sigma})
$$ 
where $\vec{r}$ is the unknown Bloch vector of the state and $\vec{\sigma} = (\sigma_{x}, \sigma_{y}, \sigma_{z})$ are the Pauli matrices in $M(\mathbb{C}^{2})$. Following the methodology of the previous section, we concentrate on the structure of the statistical model itself rather than optimal state estimation. 
The latter, and other statistical problems can be solved easily once the convergence to a Gaussian model is established.

By measuring a small proportion $n^{1-\epsilon}\ll n$ of the systems we can devise 
 an initial rough estimator $\rho_{0}:=\rho_{\vec{r}_{0}}$ so that with high probability the state is in a ball of size $n^{-1/2+\epsilon}$ around  $\rho_{0}$ \cite{Guta&Kahn}. We label the states in this ball by the local parameter $\vec{u}$
$$
\rho_{\vec{u}/\sqrt{n}} = 
\frac{1}{2} \left(
\mathbf{1} + (\vec{r}_{0} + \vec{u}/\sqrt{n})\vec{\sigma}\right)
$$
and define the local statistical model by
\begin{equation}\label{eq.q.n}
\mathcal{Q}_{n}:= \left\{\rho^{n}_{\vec{u}} : \| \vec{u}\|\leq n^{\epsilon} \right\} ,\qquad \rho^{n}_{\vec{u}} :=\rho^{\otimes n}_{\vec{u}/\sqrt{n}}. 
\end{equation}

By choosing a  coordinate system $(\vec{a}_{1},\vec{a}_{2},\vec{a}_{3})$ with 
$\vec{a}_{3}$ along  $\vec{r}_{0}$  and writing 
$\vec{u}= u_{1}\vec{a}_{1}+u_{2}\vec{a}_{2}+u_{3}\vec{a}_{3}$ we observe that  
$\rho_{\vec{u}/\sqrt{n}}$ is essentially obtained by perturbing the eigenvalues of 
$\rho_{0}$ by $u_{3}/2\sqrt{n}$ and rotating it with a `small' unitary
$$
U:= \exp(i(- u_{2} \vec{a}_{1}+u_{1}\vec{a}_{2})\vec{\sigma}/2r_{0}\sqrt{n}),\qquad
r_{0}:=\|\vec{r}_{0}\|.
$$ 

The splitting into `classical' and `quantum' parameters $u_{3}$ and $(u_{1},u_{2})$ can be intuitively explained through the `big Bloch sphere' picture commonly used to describe spin coherent \cite{Radcliffe} and spin squeezed states \cite{Kitagawa&Ueda}. 
Let
$$
L_{j}:= \sum_{i=1}^{n}  \vec{a}_{j} \cdot \vec{\sigma}^{(i)}  ,\qquad j=1,2,3
$$
be the collective spin components along the directions $\vec{a}_{j}$. 
By the Central Limit Theorem, the distributions of $L_{i}$ with respect to 
$\rho_{0}^{\otimes n}$ converge as
$$
\frac{1}{\sqrt{n}}(L_{3} - n r_{0}) \overset{\mathcal{D}}{\longrightarrow}
N(0, 1-r_{0}^{2}),\qquad
\frac{1}{\sqrt{n}}L_{1,2}  \overset{\mathcal{D}}{\longrightarrow}  N(0, 1), 
$$
so that the joint spins state can be pictured as a vector of length $nr_{0}$ 
whose tip has a Gaussian blob of size $\sqrt{n}$ 
representing the uncertainty in the collective variables (see Figure \ref{fig.big.ball}). Furthermore, by a law of large numbers heuristic we estimate the commutators 
$$
\left[\frac{1}{\sqrt{n}} L_{1}, \frac{1}{\sqrt{n}} L_{2}\right]=
2i \frac{1}{n} L_{3} \approx
 2i r_{0} \mathbf{1},\qquad
\left[\frac{1}{\sqrt{n}} L_{1,2}, \frac{1}{\sqrt{n}} L_{3}\right] 
\approx 0.
$$
\begin{figure}
\label{fig.big.ball}
\begin{center}
\includegraphics[width=6cm]{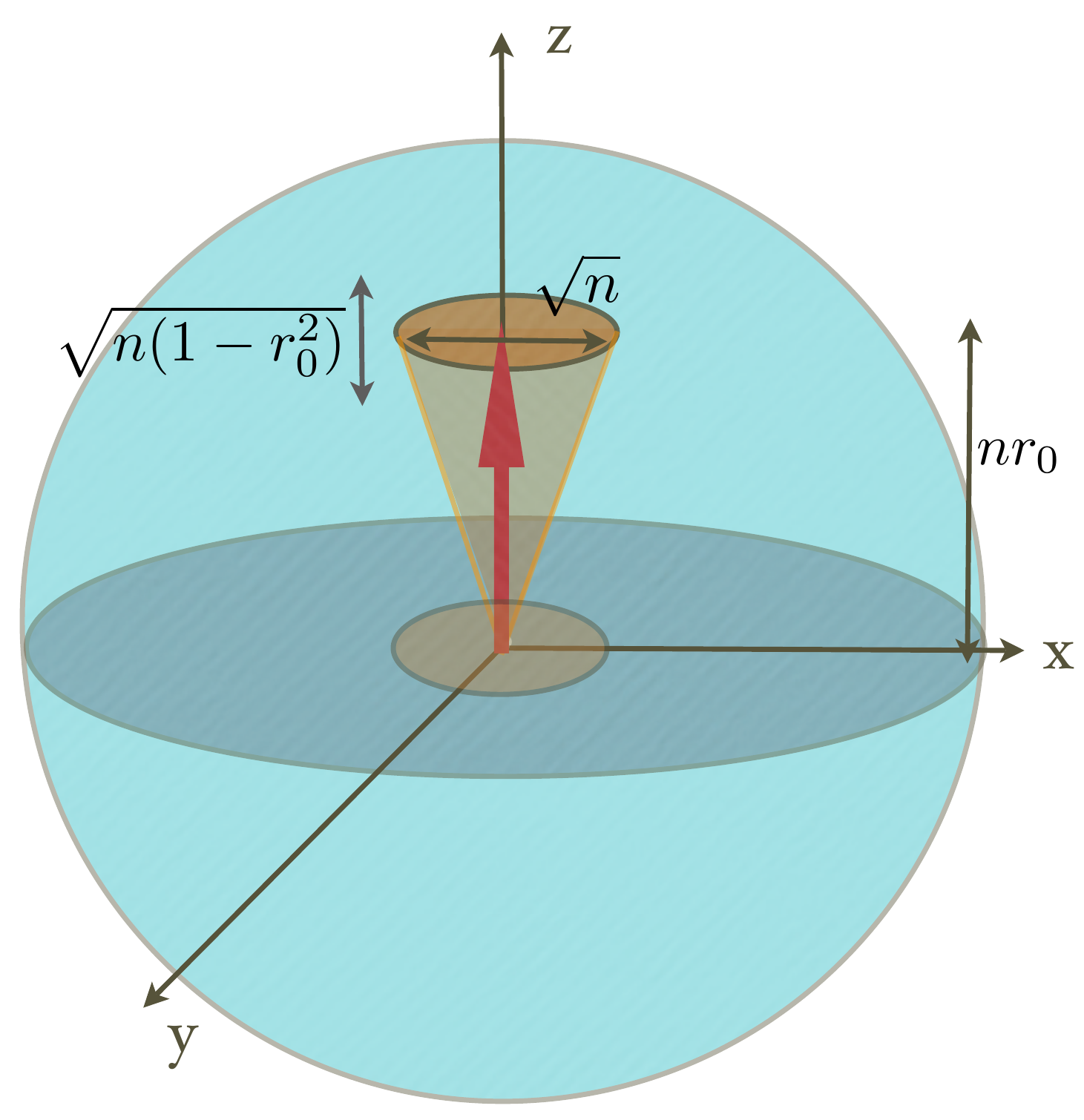}
\caption{Big ball picture of the collective state of identical mixed spins. The total spin  is represented as a vector of length $nr_{0}$ with a 3D uncertainty blob of size 
$\sqrt{n}$ in the $x,y$ directions and $\sqrt{n(1-r_{0}^{2})}$ in the $z$ direction.}
\end{center}
\end{figure}

This suggests that $ L_{1}/\sqrt{2r_{0}n}$ and 
$L_{2}/\sqrt{2r_{0}n}$ converge to the canonical coordinates $Q$  and $P$ of a quantum harmonic oscillator in a thermal equilibrium state 
$$
\Phi:= (1-p) \sum_{k=0}^{\infty} p^{k} | k\rangle\langle k|,\qquad 
p= \frac{1-r_{0}}{1+r_{0}},
$$
where $\{|k\rangle :k\geq 0\}$ represents the Fock basis. 
Moreover  the (rescaled) component $\frac{1}{\sqrt{n}}(L_{3} - n r_{0})$ converges to a {\it classical} Gaussian variable $X\sim N:= N(0,1-r_{0}^{2})$ which is independent of the quantum state. Note that the Gaussian limit state has both quantum and classical components and should be identified with the state 
$\Phi\otimes N $ on the von Neumann algebra $\mathcal{B}(\ell^{2}(\mathbb{N})) \otimes L^{\infty}(\mathbb{R})$.

What is the Gaussian state when the spins are in the `perturbed' state 
$\rho^{n}_{\vec{u}}$ ? By applying the same argument we obtain that the variables $Q,P,X$ pick up expectations which (in the first order in $n^{-1/2}$) are proportional to the local parameters $(u_{1},u_{2},u_{3})$ while the variances remain unchanged. More precisely the oscillator is in a 
displaced thermal equilibrium state 
$
\Phi_{\vec{u}} := D(\vec{u}) \Phi D(\vec{u})^{*} ,
$
where $D(\vec{u})$ is the displacement operator 
$$
D(\vec{u}):=\exp \left( i (- u_{2} Q+u_{1}P )/\sqrt{2r_{0}} \right),
$$ 
and the classical bit has distribution $N_{\vec{u}}:=N(u_{3}, 1-r_{0}^{2})$.

\begin{definition}\label{def.quantum.gaussian.shift}
The quantum Gaussian shift model $\mathcal{G}$ is defined by the family of quantum-classical states
\begin{equation}\label{eq.q.Gaussian.shift}
\mathcal{G}:= \{ \Phi_{\vec{u}}\otimes N_{\vec{u}} : \vec{u}\in\mathbb{R}^{3} \}
\end{equation}
on $\mathcal{B}(\ell^{2}(\mathbb{N})) \otimes L^{\infty}(\mathbb{R})$.
\end{definition}

Having defined the sequence of local models $\mathcal{Q}_{n}$ and the Gaussian shift model, we need to define the quantum counterparts of randomisations 
and convergence of models. The natural analogue of a classical randomisation is a quantum channel, i.e. completely positive, trace preserving map 
$
C:\mathcal{T}_{1}(\mathcal{H})\to \mathcal{T}_{1}(\mathcal{K})
$
where $\mathcal{T}_{1}(\mathcal{H})$ represents the trace class operators on $\mathcal{H}$. However, as we saw above, a sequence of quantum statistical 
models may converge to a quantum-classical one. The mathematical framework covering randomisations of both classical and quantum statistical models is that of von Neuman algebras and channels between their preduals. In finite dimensions this simply means that we deal with channels between block diagonal matrix algebras. We can now define the Le Cam distance between two quantum models in the same way as in definition \ref{def.LeCam.distance} with classical randomisation replaced by quantum ones and the $\|\cdot\|_{1}$ representing the norm on the predual, which is the trace norm in the case of density matrices.

 \begin{theorem}\label{th.qlan.qubits}
 Let $\mathcal{Q}_{n}$ be the sequence of statistical models \eqref{eq.q.n} for $n$ i.i.d. local spin states.
and let $\mathcal{G}_{n}$ be the restriction of the Gaussian shift model 
\eqref{eq.q.Gaussian.shift} to the range of parameters $\|\vec{u}\|\leq n^{\epsilon}$.
Then
$$
\lim_{n\to\infty}\Delta(\mathcal{Q}_{n}, \mathcal{G}_{n}) =0,
$$  
i.e. there exist sequences of channels $T_{n}$ and $S_{n}$ such that

\begin{equation}\label{eq.channel.conv.}
\begin{split}
\lim_{n\to \infty}\, 
\sup_{\| {\bf u}\|\leq n^{\epsilon}} 
\| \Phi_{ \vec{u}}\otimes N_{\vec{u}}  - 
T_{n} \left(  \rho_{ \vec{u}}^{n}\right)  \|_{1} =0, \\
\lim_{n\to \infty} \,
\sup_{\| {\bf u}\|\leq n^{\epsilon}} \| \rho_{\vec{u}}^{n} - S_{n} \left(  \Phi_{ \vec{u}}\otimes N_{\vec{u}} \right)  \|_{1} =0. \\
\end{split}
\end{equation}
\end{theorem}

To conclude this section we would like to make a few comments on the significance of the above result. The first point is that although it was intuitively illustrated using the Central Limit Theorem, the concept of local asymptotic normality provides a stronger characterisation of the `Gaussian approximation'. 
Indeed the convergence in Theorem \ref{th.qlan.qubits} is strong (in $L_{1}$) rather than weak (in distribution), it is {\it uniform} over a range of local parameters rather than at a single point, and has an operational meaning based on quantum channels.

Secondly, one can exploit these features to devise asymptotically optimal measurement strategies for state estimation and prove that the Holevo bound 
\cite{Holevo} is asymptotically attainable \cite{Guta&Kahn3}.

Thirdly, the result can be applied to other quantum statistical problems 
involving i.i.d. qubit states such as cloning, teleportation benchmarks, quantum learning, and can serve as a mathematical framework for analysing quantum state transfer protocols.

\section{Local formulation of the classification problem}\label{sec.local.formulation}

In this section we reformulate the problem of quantum state classification in the `local' set-up. This allows us to replace, on the one hand the excess error probability by a {\it quadratic} form in local parameters, and on the other hand 
the training set consisting of i.i.d. spins by a simpler Gaussian shift model.

Throughout the section we restrict to the case where the priors 
$\pi_{0},\pi_{1}$ are known. In Section \ref{sec.priors} we show that the results for known priors can easily be extended to unknown ones by simply estimating them from the counts of $\rho$ and $\sigma$ states in the training sample.


\subsection{The loss function}

Recall that the classification problem is to discriminate between two 
unknown states  $\rho$ and $\sigma$ by learning from a training set of $n$ {\it labelled} systems prepared randomly in one of the states with probabilities 
$\pi_{0}$ and $\pi_{1}$. For this we measure the training set and produce an outcome which is itself a measurement 
$\widehat{M}_{n}:=(\widehat{P}_{n},\mathbf{1}-\widehat{P}_{n})$ on $\mathbb{C}^{2}$. The accuracy of the procedure 
is measured by the excess risk (\ref{eq.quantum_excess}):
\begin{equation}\label{eq.excess.risk}
R(\widehat{M}_{n}) =\mathbb{E} {\rm Tr}\left[(\pi_1 \sigma - \pi_0 \rho) (\widehat{P}_{n} - P^*) \right],
\end{equation}
with $P^*=[\pi_{0}\rho-\pi_{1}\sigma]_{+}$. Since any binary measurement is a mixture of projective POVM's \cite{DAriano_et_al2005}, we can assume without loss of generality 
that $\widehat{P}_{n}$ is a projection and pull back the randomness into the definition of the training set measurement.

As explained in section \ref{sec.qlan}, the a priori unknown states $\rho$ and $\sigma$ can be localised with high probability in $n^{-1/2+\epsilon}$ neighbourhoods of $\rho_{0}$ and $\sigma_{0}$ by sacrificing a small proportion of the training set systems; this means that $\rho_0$ and $\sigma_0$ are known and can be used by the classification procedure. 
Let $\vec{r}_{0}$ and $\vec{s}_{0}$ be the Bloch vectors of $\rho_{0}$ and 
$\sigma_{0}$ and let us parametrise their neighbourhoods as follows
\begin{align}
&\rho = \rho_{\vec{u}/\sqrt{n}} = 
\frac{1}{2} \left(
\mathbf{1} + \left(\vec{r}_{0}+ \frac{\vec{u}}{\sqrt{n}}\right)\vec{\sigma}
\right),\nonumber\\
&  \sigma = \sigma_{\vec{v}/\sqrt{n}} = 
\frac{1}{2} \left(
\mathbf{1} + \left(\vec{s}_{0}+ \frac{\vec{v}}{\sqrt{n}}\right)\vec{\sigma} \right).
\label{eq.rhosigma.local}
\end{align}

Let $P_{0}:= [\pi_{0}\rho_{0}-\pi_{1}\sigma_{0}]_{+}$ be the optimal projection corresponding to the pair $(\rho_{0},\sigma_{0})$ and note that it can have dimension one, or it can be zero or identity. 
In the second case, the optimal measurement is {\it trivial}, one can guess the state without measuring by checking whether the operator $\pi_{0}\rho_{0}-\pi_{1}\sigma_{0}$ is positive or negative. 
\begin{lemma}
Let $(\rho_{0},\sigma_{0})$ and $(\pi_{0}, \pi_{1})$ satisfy 
$$
\|\pi_{0}\vec{r}_{0} -\pi_{1}\vec{s}_{0}\| < |\pi_{0} - \pi_{1}|.
$$
Then $P_{0}$ is either zero or identity and the local minimax excess risk satisfies
$$
\inf_{\widehat{M}_{n}}\sup_{\|\vec{u} \|,\|\vec{v}\|\leq n^{\epsilon}} 
\mathbb{P}_{e}(\widehat{M}_{n}) - P_{e}^{*}  =O(\exp(-cn))
$$
for some $c>0$.
\end{lemma}

\begin{proof}
Note that the inequality is satisfied only if $\pi_{0}\neq \pi_{1}$ and it implies that
$$
\pi_{0}\rho_{0} - \pi_{1}\sigma_{0}= 
\frac{\pi_{0}-\pi_{1}}{2}
\left(\mathbf{1} + 
\frac{\pi_{0}\vec{r}_{0} -\pi_{1}\vec{s}_{0}}{\pi_{0}-\pi_{1}} \vec{\sigma}
\right)
$$
it a positive or negative operator depending on the sign of $\pi_{0}-\pi_{1}$.

Since both eigenvalues of $\pi_{0}\rho_{0} - \pi_{1}\sigma_{0}$ are non-zero, 
there exists a constant $\eta>0$ such that 
$$
\|\pi_{0}\rho_{0} - \pi_{1}\sigma_{0} -A\|_{2} \leq \eta
$$
implies that $A$ is also  a positive or negative operator. In fact, when $n$ is 
large enough all $\pi_{0}\rho_{\vec{u}/\sqrt{n}} - \pi_{1}\sigma_{\vec{v}/\sqrt{n}}$ with $\|\vec{u}\|,\|\vec{v}\|\leq n^{\epsilon}$ have this property for some other constant $\tilde{\eta}$. 

Consider a simple measurement on the training set where the states are measured separately in the three bases of the Pauli matrices and the outcomes averages are used to construct a estimators of the states $\rho_{\vec{u}/\sqrt{n}}$ and 
$\sigma_{\vec{v}/\sqrt{n}}$. Then by basic concentration inequalities we get
$$
\mathbb{P}\left( 
\left\| 
\pi_{0}\left( \rho_{\vec{u}/\sqrt{n}}- \rho_{\hat{\vec{u}}/\sqrt{n}}\right) + 
\pi_{1} \left(\sigma_{\vec{v}/\sqrt{n}}- \sigma_{\hat{\vec{v}}/\sqrt{n}}\right)
\right\|_{2} \geq \tilde{\eta}\right)\leq exp(-cn) 
$$
which means that with exponentially small probability error the plug-in estimator of 
$P^{*}:= [\pi_{0}\rho_{\vec{u}/\sqrt{n}} - \pi_{1}\sigma_{\vec{v}/\sqrt{n}}]_{+}$ will be equal to $P^{*}$ which is zero or identity. 

\end{proof}

From now on we will work under the assumption that 
\begin{equation}\label{eq.assumption}
\|\pi_{0}\vec{r}_{0}-\pi_{1}\vec{s}_{0}\|>|\pi_{0}-\pi_{1}|,
\end{equation}
so that $P_{0}:=[\pi_{0}\rho_{0} - \pi_{1}\sigma_{0}]_{+}$ is a one dimensional projection whose Bloch vector is 
$$
\vec{p}_{0} 
= 
\frac{\vec{d}_{0}}{\|\vec{d}_{0}\|} 
:= 
\frac{ \pi_{0} \vec{r}_{0}  - \pi_{1} \vec{s}_{0} }{\| \pi_{0}\vec{r}_{0} - \pi_{1} \vec{s}_{0}  \| }.
$$


The Helstrom projection $P^{*}$ for the pair of unknown states $(\rho,\sigma)$ has Bloch vector 
\begin{equation}
\vec{p} 
= \frac{\vec{d}}{\|\vec{d}\|} = 
\frac{ \pi_{0}  \left(\vec{r}_{0} + \frac{\vec{u}}{\sqrt{n}}\right) -\pi_{1} \left(\vec{s}_{0} + \frac{\vec{v}}{\sqrt{n}}\right) }
{\left\| \pi_{0} \left(\vec{r}_{0} + \frac{\vec{u}}{\sqrt{n}}\right) - \pi_{1} \left(\vec{s}_{0} + \frac{\vec{v}}{\sqrt{n}}\right) \right\| }
= \frac{\vec{d}_{0} + \frac{\vec{z}}{\sqrt{n}}}{\left\|\vec{d}_{0} + \frac{\vec{z}}{\sqrt{n}}\right\|},
\label{eq.oracle.bloch}
\end{equation}
where $\vec{z}:=\pi_{0}\vec{u}-\pi_{1}\vec{v} $ is a relative parameter and 
$\vec{d} :=\vec{d}_{0}+ \frac{\vec{z}}{\sqrt{n}}$.

As discussed before, we can take the estimator $\widehat{M}_{n}$ to be a projective measurement 
$\widehat{M}_{n}:=(\widehat{P}_{n},\mathbf{1}-\widehat{P}_{n})$, so to minimise the risk \eqref{eq.excess.risk} we aim at producing an estimator 
$\widehat{P}_{n}$ which is close to $P^{*}$. Since the latter is obtained by rotating $P_{0}$ with angle of order $n^{-1/2+\epsilon}$, we can assume without loss of generality that $\widehat{P}_{n}$ has a Bloch vector $\hat{\vec{p}}_{n}$ which is a 
small rotation of $\vec{p}_{0}$ so that 
  \begin{equation}\label{eq.estimator.bloch}
\hat{\vec{p}}_{n}= \frac{\vec{p}_{0} + \hat{\vec{z}}_{n}/\sqrt{n}}{\|\vec{p}_{0} + \hat{\vec{z}}_{n}/\sqrt{n}\|},
\end{equation}
with $\hat{\vec{z}}_{n}=O(n^{\epsilon})$ 
a vector in the plane orthogonal to $\vec{p}_{0}$. 

Expanding \eqref{eq.oracle.bloch} and \eqref{eq.estimator.bloch} in powers of $n^{-1/2}$ we get
\begin{align*}
\vec{p}-\hat{\vec{p}}_{n}
&= 
\frac{1}{\sqrt{n}} 
\left[ 
\frac{\vec{z} -\hat{\vec{z}}_n }{\|\vec{d}_{0}\|} - 
\frac{\vec{d}_{0} (\vec{d}_{0} \cdot (\vec{z} -\hat{\vec{z}}_n) ) }{\|\vec{d}_{0}\|^{3}}
\right]\\
&+
\frac{1}{n}\left[
 -\frac{\vec{d}_{0}(\|\vec{z}\|^{2}- \|\hat{\vec{z}}_n\|^{2} )}{2\|\vec{d}_{0}\|^{3} } + \frac{3 \vec{d}_{0} 
(( \vec{d}_{0}\cdot \vec{z} )^{2} - (\vec{d}_{0} \cdot \hat{\vec{z}}_n)^{2} }{2\|\vec{d}_{0}\|^{5}} \right] +o(n^{-1}).
\end{align*}

%

We now plug these expressions back into into 
\eqref{eq.excess.risk} taking into account that $\hat{\vec{z}}_n$ is perpendicular to $\vec{d}_{0}$ and obtain
\begin{eqnarray*}
\mathbb{P}_{e}(\widehat{M}_{n})-P^{*}_{e} 
&= &\mathbb{E} {\rm Tr}\left( (\pi_{0}\rho-\pi_{1}\sigma) (P-\widehat{P}_{n}) \right)\\
&=&
\frac{1}{2} \mathbb{E}\vec{d} \cdot (\vec{p}-\hat{\vec{p}}_{n})\\
&=&
\frac{1}{4n \|\vec{d}_{0} \|} \mathbb{E}\| \vec{z}_{\perp} -\hat{\vec{z}}_n\|^{2} + o(n^{-1})
\end{eqnarray*}
where $\vec{z}_{\perp}= \vec{z} - \vec{d}_{0}(\vec{z}\cdot \vec{d}_{0})/\|\vec{d}_{0}\|^{2} $ is the projection of $\vec{z}$ onto the plane orthogonal to $\vec{d}_{0}$. 

It is clear now that the rate of convergence of the excess risk \eqref{eq.excess.risk} 
is $n^{-1}$, so it  is meaningful to optimise the quantity 
$n R^{(l)}_{max}(\widehat{M}_{n})$, and the contribution coming from the $o(n^{-1})$ term can be dropped. 

Since $\widehat{M}_{n}$ is uniquely determined by $\hat{\vec{z}}_{n}$ by 
\eqref{eq.estimator.bloch}, we define the {\it quadratic} loss function for the measurement on the training set in terms of local variables
\begin{equation}\label{loss.fct.quadratic}
L((\vec{u},\vec{v}), \hat{\vec{z}}_n):= 
\frac{1}{4 \|\vec{d}_{0} \|} \| \vec{z}_{\perp} -\hat{\vec{z}}_n\|^{2},   \quad 
\vec{z}:=(\pi_{0}\vec{u}-\pi_{1}\vec{v})
\end{equation}
and the associated renormalised risk is
$
R_{\vec{u},\vec{v}}(\hat{\vec{z}}_n):= \mathbb{E} L((\vec{u},\vec{v}), \hat{\vec{z}}_n).
$
The local maximum risk \eqref{eq.local.maximum.risk} 
around $(\rho_{0},\sigma_{0})$ is then
\begin{align}\label{eq.quadratic.max.risk}
R^{(l)}_{max} (\hat{\vec{z}}_{n}~; ~\rho_{0},\sigma_{0})
&:=
\sup_{\|\vec{u}\|,\|\vec{v}\|\leq n^{\epsilon}} 
R_{\vec{u},\vec{v}}(\hat{\vec{z}}_n)\nonumber\\
&=
\sup_{\|\vec{u}\|,\|\vec{v}\|\leq n^{\epsilon}} \frac{1}{4 \|\vec{d}_{0} \|} \mathbb{E}\| \vec{z}_{\perp} -\hat{\vec{z}}_n\|^{2}.
\end{align}
In conclusion, we need to find the optimal measurement strategy on the training 
set with respect to the above quadratic form of the local parameters.

\subsection{The training set}

To solve the above problem we employ the machinery of local asymptotic normality. As before, let $\rho$ and $\sigma$ be states in local neighbourhood of $\rho_{0}$ and respectively $\sigma_{0}$ described by \eqref{eq.rhosigma.local}. We write their local Bloch vectors $(\vec{u},\vec{v})$ as
$$
\vec{u} = u_{1} \vec{a_{1}} + u_{2} \vec{a_{2}} + u_{2}\vec{a_{3}} 
\qquad{\rm and}\qquad 
\vec{v} = v_{1} \vec{b_{1}} + v_{2} \vec{b_{2}} + v_{2}\vec{b_{3}}
$$ 
where $(\vec{a}_{1},\vec{a}_{2},\vec{a}_{3})$ and 
$(\vec{b}_{1},\vec{b}_{2},\vec{b}_{3})$ are two coordinate systems which satisfy the conditions (see Figure \ref{fig.bloch2})
\begin{enumerate}
\setlength{\itemsep}{0pt}
\setlength{\parskip}{1pt}
\setlength{\parsep}{1pt}
\item 
$\vec{a}_{3}$ is parallel to $\vec{r}_{0}$, 
\item
$\vec{b}_{3}$ is parallel to $\vec{s}_{0}$ ,
\item 
$\vec{a}_{1},\vec{b}_{1}$ are in the plane $(\vec{r}_{0},\vec{s}_{0})$,
\item
$\vec{a}_{2}= \vec{b}_{2}$ is perpendicular to the plane  $(\vec{r},\vec{s})$.
\end{enumerate}

With these notations the  local statistical model for the training set is
$$
\mathcal{T}_{n}:= \{ \rho_{\vec{u}}^{n\pi_{0}} \otimes \sigma_{\vec{v}}^{n\pi_{1}} :\|\vec{u}\|, \|\vec{v}\|\leq n^{\epsilon}\}
$$
and the corresponding Gaussian shift model is 
\begin{equation}\label{eq.g2}
\mathcal{G}^{(2)}:= \{ N_{\vec{u}} \otimes N_{\vec{v}} \otimes \Phi_{\vec{u}} \otimes \Phi_{\vec{v}} : \vec{u}, \vec{v}\in \mathbb{R}^{3} \}
\end{equation}
where 
\begin{align}
N_{\vec{u}} &:= N(\sqrt{\pi_{0}} u_{3}, 1-r_{0}^{2} ),
\nonumber\\ 
N_{\vec{v}} &:= N(\sqrt{\pi_{1}} v_{3}, 1-s_{0}^{2}),
\nonumber \\
\Phi_{\vec{u}}&:= \Phi\left(\sqrt{\frac{\pi_{0}}{2r_{0}}}u_{1},\,\sqrt{\frac{\pi_{0}}{ 2r_{0}}} u_{2}; \,\frac{\mathbf{1}}{2r_{0}}\right),
\nonumber\\
\Phi_{\vec{v}}&:= \Phi\left(\sqrt{\frac{\pi_{1}}{2s_{0}}}v_{1},\,\sqrt{\frac{\pi_{1}}{ 2s_{0}}} v_{2}; \,\frac{\mathbf{1}}{2s_{0}}\right)
\end{align}
and $\Phi(q,p,v)$ is a displaced thermal equilibrium state with means 
$(q,p)$ and variance $v$.

The following technical lemma shows that local asymptotic normality can be used to transfer the problem of the optimal classification from a training set consisting of qubits, to a Gaussian one. The arguments are rather standard though tedious, and since the same method has been used for finding the optimal estimation procedure for qubits \cite{Guta&Janssens&Kahn}, we refer to that paper for the proof.   

\begin{lemma}
Consider the problems of finding asymptotically optimal strategies for the models 
$\mathcal{T}_{n}$ and respectively $\mathcal{G}^{(2)}_{n}$ with respect to the 
loss function \eqref{loss.fct.quadratic}. Then the local minimax risks of both problems converge to the same constant which is the 
the minimax risk of the unrestricted Gaussian shift model $\mathcal{G}^{(2)}$.
\end{lemma}

In conclusion, the measurement of the training set should be aimed at optimally estimating the two parameter vector $\vec{z}_{\perp}$ directly, rather than using a `plug-in' strategy where the three dimensional local parameters $(\vec{u},\vec{v})$ are  first (optimally) estimated and then the measurement 
$\widehat{P}_{n}$ is  constructed as in \eqref{eq.plug.in.quantum}. We will come back to this point later on when the two methods will be compared.

\section{Optimal classifier}
\label{sec.theorem}
In this section we formulate our main result characterising the asymptotically 
optimal measurement on the training set and derive the expression of the 
optimal excess risk. Summarising the previous section, we transformed the original problem into a parameter estimation one for the Gaussian shift model 
\eqref{eq.g2}  with parameters 
$(\vec{u},\vec{v})\in\mathbb{R}^{3}\times \mathbb{R}^{3}$. 
The parameter to be estimated $\vec{z}_{\perp}\in \mathbb{R}^{2}$ is a linear transformation of $(\vec{u},\vec{v})$ 
$$
\vec{z}_{\perp}= \vec{z} - \vec{d}_{0}(\vec{z}\cdot \vec{d}_{0})/\|\vec{d}_{0}\|^{2} , 
\qquad \vec{z}:= \pi_{0}\vec{u}-\pi_{1}\vec{v}
$$
i.e. we would like to minimise the risk
$$
R_{max}(\hat{\vec{z}}; \rho_{0},\sigma_{0}) := 
\sup_{\vec{u},\vec{v}}\mathbb{E} L((\vec{u},\vec{v}), \hat{\vec{z}})=
\sup_{\vec{u},\vec{v}} \frac{1}{4\|\vec{d}_{0}\|} 
\mathbb{E}\|\hat{\vec{z}}- \vec{z}_{\perp}\|^{2}. 
$$.

Since the local parameters contain both classical and quantum components it is convenient to express the loss function $L((\vec{u},\vec{v}), \hat{\vec{z}})$ in terms 
of these components. Let $(\vec{p}_{0},\vec{l}_{0},\vec{k}_{0})$ be the reference frame with  $\vec{l}_{0}$ in the plane $(\vec{r}_{0},\vec{s}_{0})$. Denote by $\varphi_{0},\varphi_{1}$ the angles between $(\vec{r}_{0},\vec{l}_{0})$ and  respectively 
$(\vec{s}_{0},\vec{l}_{0})$ (see Figure \ref{fig.bloch2}). 
Then $\vec{z}_{\perp}= z_{l} \vec{l}_{0}+z_{k} \vec{k}_{0}$ with components 
\begin{align*}
z_{l}
&= 
(\pi_{0}\cos\varphi_{0} u_{3}- \pi_{1}\cos\varphi_{1}v_{3})+(\pi_{0}\sin\varphi_{0} u_{1} +\pi_{1}\sin\varphi_{1} v_{1} )
:= z^{(c)}_{l}+z^{(q)}_{l},
\\
z_{k}&= \pi_{0} u_{2} - \pi_{1} v_{2} 
\end{align*}
where $z_{l}$ was split into a contribution coming from the `classical' parameters 
$(u_{3},v_{3})$, and another one from the `quantum' parameters. Since the classical and quantum parts of the Gaussian model are {\it independent} it is easy to verify that the optimal estimator $\hat{\vec{z}}$ can be written as 
$$
\hat{\vec{z}} = (\hat{z}^{(c)}_{l}+\hat{z}^{(q)}_{l})\vec{l}_{0}+ \hat{z}_{k} \vec{k}_{0}
$$
where $\hat{z}^{(c)}_{l}$ is the optimal estimator of  $z^{(c)}_{l}$ and 
$(\hat{z}^{(q)}_{l},\hat{z}_{k})$ are optimal estimators of 
$(z^{(q)}_{l},z_{k})$ obtained by (jointly) measuring the two quantum Gaussian components. The excess risk can be written as
\begin{eqnarray*}
4\|\vec{d}_{0}\| \mathbb{E}\left[ L\left((\vec{u},\vec{v}),\hat{\vec{z}}\right)\right]&=& 
\mathbb{E}\left[ (z^{(c)}_{l} - \hat{z}^{(c)}_{l} )^{2}\right]  \\
&+&
\mathbb{E}\left[ (z^{(q)}_{l} - \hat{z}^{(q)}_{l})^{2} + (z_{k} - \hat{z}_{k})^{2}\right] 
\end{eqnarray*}
where the classical and quantum contributions separate and can be optimised separately. The optimal choice for the classical estimator is 
$$
\hat{z}^{(c)}_{l} = 
\sqrt{\pi_{0}} \cos\varphi_{0} X_{r}  -  
\sqrt{\pi_{1}} \cos\varphi_{1} X_{s}
$$
where $(X_{r},X_{s})\sim N_{\vec{u}}\otimes N_{\vec{v}}$ denote the random variables making up the classical part of the limit Gaussian model. Its contribution to the excess risk is
\begin{equation}\label{eq.classical.risk}
\mathbb{E}\left[ \left(z^{(c)}_{l} - \hat{z}^{(c)}_{l} \right)^{2}\right] = 
\pi_{0} (1-r_{0}^{2})\cos^{2}\varphi_{0} + \pi_{1} (1-s_{0}^{2})\cos^{2}\varphi_{1}. 
\end{equation}
On the other hand $(z^{(q)}_{l}, z_{k})$ are the means of the canonical coordinates 
\begin{align}
Q^{(l)}&
:= \sqrt{2r_{0}\pi_{0}} \sin\varphi_{0} Q_{1} + \sqrt{2s_{0}\pi_{1}} \sin\varphi_{1}Q_{2},\nonumber\\
Q^{(k)}&
:= \sqrt{2r_{0}\pi_{0}}P_{1} - \sqrt{2s_{0}\pi_{1}} P_{2}
\label{eq.ql.qk}
\end{align} 
whose commutator is 
$$
[Q^{(l)},Q^{(k)}]= i(2r_{0}\pi_{0}\sin\varphi_{0}-2s_{0}\pi_{1}\sin\varphi_{1}) \mathbf{1}:=ic\mathbf{1}.
$$
Now, the optimal joint measurement of canonical variables is the heterodyne type 
where the non-commuting coordinates are combined with the coordinates of an additional oscillator prepared in a squeezed state 
\cite{Holevo,Guta&Janssens&Kahn}. The optimal mean square error is
\begin{align}
\mathbb{E}\left[ (z^{(q)}_{l} - \hat{z}^{(q)}_{l})^{2} + (z_{k} - \hat{z}_{k})^{2}\right] 
&= Var(Q^{(l)})+Var(Q^{(k)}) +|c| 
 \nonumber \\
& = \pi_{0}\sin^{2}\varphi_{0}+\pi_{1}\sin^{2}\varphi_{1} +1 \nonumber\\
& + 2 |\pi_{0}r_{0}\sin\varphi_{0}- \pi_{1}s_{0}\sin\varphi_{1}|
\label{eq.quantum.risk}
\end{align}
Adding the classical and quantum contributions \eqref{eq.classical.risk} and 
\eqref{eq.quantum.risk} we obtain the minimax risk
\begin{equation}\label{eq.optimal.risk}
R_{minmax}^{(l)}(\rho_{0},\sigma_{0})= 
\frac {2+2|\pi_{0}r_{0}\sin\varphi_{0}- \pi_{1}s_{0}\sin\varphi_{1}| -r_{0}s_{0} \cos\varphi_{0} \cos\varphi_{1}}{4\|\vec{d}_{0}\|}
\end{equation}
which only depends on the states $(\rho_{0},\sigma_{0})$, for given priors $(\pi_{0},\pi_{1})$.

\begin{theorem}\label{th.main}
Consider the quantum classification problem with training set 
$\rho^{\otimes \pi_{0}n}\otimes \sigma^{\otimes \pi_{1} n}$ where 
$\rho,\sigma$ are unknown qubit states and $(\pi_{0},\pi_{1})$ are known. 

Let $R_{minmax}^{(l)}(\rho_{0}, \sigma_{0})$ be the local minimax risk as defined 
in Section \ref{sec.local.minimax.}. Under the assumption 
\eqref{eq.assumption}, $R_{minmax}^{(l)}(\rho_{0}, \sigma_{0})$ is given by 
\eqref{eq.optimal.risk}.

\vspace{1mm}

The optimal measurement consists of the following steps:
\begin{enumerate}
\item 
construct rough estimators of $\rho$ and $\sigma$ by measuring $n^{1-\epsilon}$ systems;
\item
transfer the localised spins state by $T_{n}$ as in Theorem \ref{th.qlan.qubits} ;
\item
perform the optimal coherent measurement of $(Q^{(l)},Q^{(k)})$ and combine with classical estimator $\hat{z}_{c}^{l}$ to produce estimator $\widehat{P}_{n}$.
\end{enumerate}
\end{theorem}

\begin{figure}
\begin{center}
\includegraphics[width=10cm]{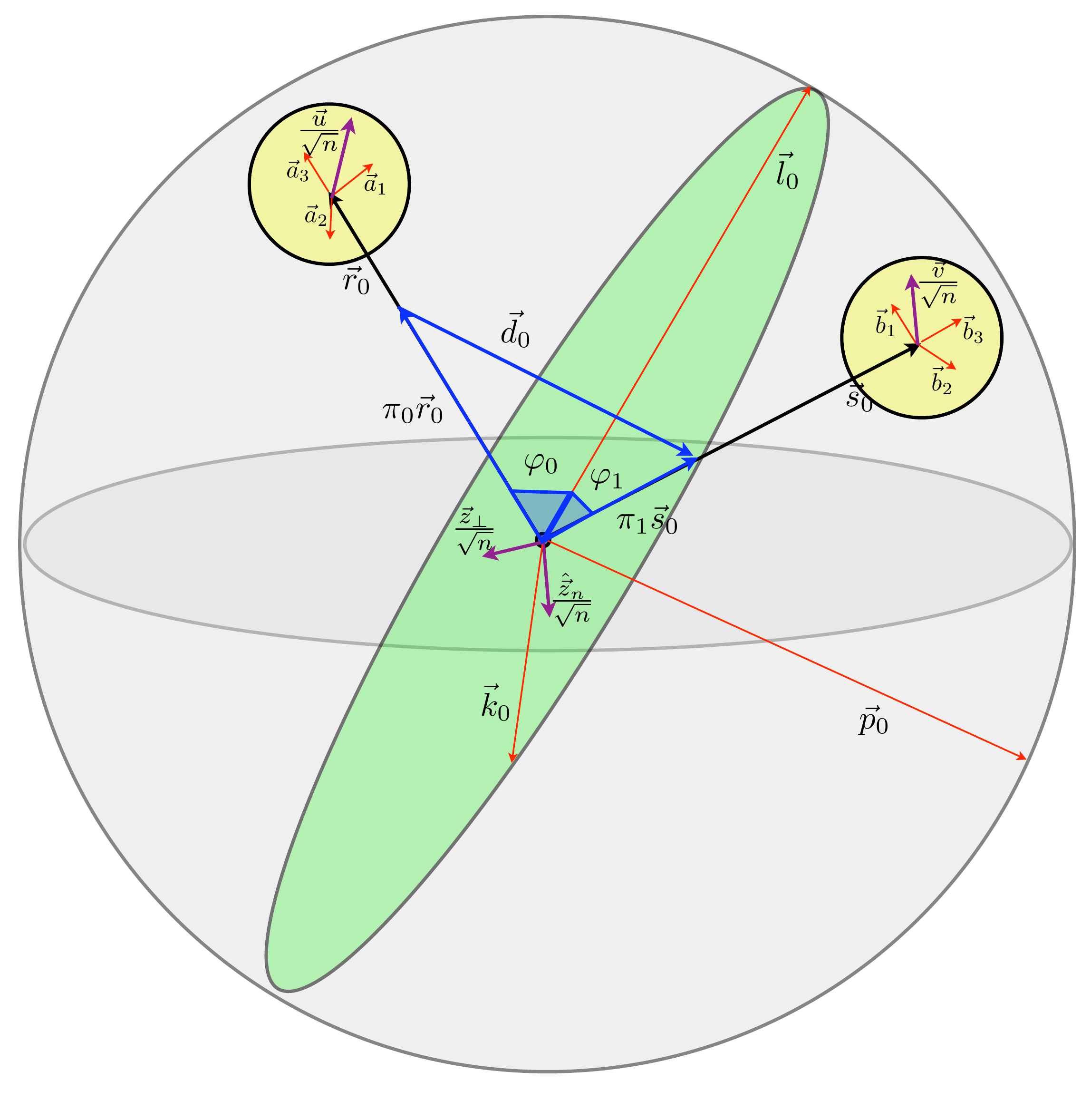}
\end{center}
\begin{flushleft}
\caption{
Bloch ball geometry of the learning problem. 

\vspace{1mm}

\noindent
The unknown states are localised in the two yellow balls centred at $\vec{r}_{0}$ and $\vec{s}_{0}$ and have local vectors $\vec{u}/\sqrt{n}$ and $\vec{v}/\sqrt{n}$ coloured in purple.\\

\noindent
The three reference systems $(\vec{a}_{1},\vec{a}_{2},\vec{a}_{3})$, $(\vec{b}_{1},\vec{b}_{2},\vec{b}_{3})$ and $(\vec{p}_{0},\vec{l}_{0},\vec{k}_{0})$ are coloured in red.\\

\noindent
The green equatorial plane is orthogonal to $\vec{p}_{0}$ and contains the estimator 
$\hat{\vec{z}}_{n}$ and the vector to be estimated $\vec{z}_{\perp}$ 
(coloured in purple).

}

\label{fig.bloch2}
\end{flushleft}
\end{figure}

\subsection{Plug-in classifier based on optimal state estimation}

\label{sec.comparison.plugin}
Here we compute the asymptotics of the renormalised risk of the plug-in 
classifier based on optimal state estimation.

The problem of optimal state estimation for mixed i.i.d. qubits was solved in the 
asymptotic local minimax setting in \cite{Guta&Janssens&Kahn}. 
The optimal measurement procedure is adaptive and the first two steps are identical to those of Theorem \ref{th.main}

\begin{enumerate}
\item 
construct rough estimators of $\rho$ and $\sigma$ by measuring $n^{1-\epsilon}$ systems;
\item
transfer the localised spins state by $T_{n}$ as in Theorem \ref{th.qlan.qubits} ;
\item
Perform separate heterodyne measurements on the modes 
$(Q_{1},P_{1})$ and $(Q_{2},P_{2})$ and observe the classical components to obtain the estimators $\tilde{\vec{u}}_{n}$ and $\tilde{\vec{v}}_{n}$.
\end{enumerate}
 
Once the states (local parameters) have been estimated we can classify new states by applying the plug-in measurement 
$\widetilde{M}_{n}:=(\widetilde{P}_{n}, \mathbf{1}-\widetilde{P}_{n})$ where 
$\widetilde{P}_{n}$ has Bloch vector
\begin{equation}\label{eq.plug.in.estimation}
\tilde{\vec{p}}= 
 \frac{\tilde{\vec{d}}}{\|\tilde{\vec{d}}\|} 
= \frac{\vec{d}_{0} + \frac{\tilde{\vec{z}}_{n}}{\sqrt{n}}}{\left\|\vec{d}_{0} + \frac{\tilde{\vec{z}}}{\sqrt{n}}\right\|}, \qquad 
\tilde{\vec{z}}_{n}:= \tilde{\vec{z}}_{\perp}:=
(\pi_{0}\tilde{\vec{u}}_{n}- \pi_{1}\tilde{\vec{v}}_{n})_{\perp}.
\end{equation}
Note that $\tilde{\vec{z}}_{n}$ was chosen to be the orthogonal component of 
$\tilde{\vec{z}}$ onto the vector $\vec{p}_{0}$ rather than $\tilde{\vec{z}}$ itself. However a simple Taylor expansion shows that the two estimators give the same leading order contribution to the risk.  As before, the minimax risk is 
the expectation of the quadratic loss function 
$L((\vec{u},\vec{v}), \tilde{\vec{z}})$ defined in 
\eqref{loss.fct.quadratic}, but now with $\tilde{\vec{z}}$ having a different 
distribution compared with  the optimal $\hat{\vec{z}}$. Again, we write 
$\tilde{z}$ as
$$
\tilde{z}=\tilde{z}_{l}\vec{l}_{0}+\tilde{z}_{k}\vec{k}_{0}=
(\tilde{z}^{c}_{l}+ \tilde{z}^{q}_{l}) \vec{l}_{0} + \tilde{z}_{k}\vec{k}_{0}
$$ 
and the risk is 
\begin{equation*}\label{eq.quadratic.max.risk.plugin}
R_{max} (\tilde{\vec{z}}~; ~\rho_{0},\sigma_{0})
=
\frac{1}{4 \|\vec{d}_{0} \|} \mathbb{E}\left[ 
(z^{(c)}_{l} - \tilde{z}^{(c)}_{l} )^{2}+
(z^{(q)}_{l} - \tilde{z}^{(q)}_{l})^{2} + 
(z_{k} - \tilde{z}_{k})^{2} \right]. 
\end{equation*}
While the contribution from the first term is given by \eqref{eq.classical.risk}, the `quantum components' have different variances due to the fact that we used a different heterodyne measurement. By using \eqref{eq.ql.qk} and the fact that heterodyne adds a factor $1/2$ to the variance of canonical coordinates we obtain
\begin{align}
\mathbb{E}\left[ (z^{(q)}_{l} - \tilde{z}^{(q)}_{l})^{2} \right]
&=
\pi_{0}\sin^{2}\varphi_{0}(r_{0}+1) + \pi_{1}\sin^{2}\varphi_{1}(s_{0}+1)
\nonumber\\
\mathbb{E}\left[ (z_{k} - \tilde{z}_{k})^{2} \right]
&=
\pi_{0}(r_{0}+1) + \pi_{1}(s_{0}+1)
\end{align}
Adding the three contributions we get
\begin{align}\label{eq.risk.plugin.estimation}
4\|d_{0}\|R_{max} (\tilde{\vec{z}}~; ~\rho_{0},\sigma_{0})
&=
 2+ \pi_{0}(r_{0}\sin^{2}\varphi_{0}+ r_{0}- r_{0}^{2}\cos^{2}\varphi_{0})\nonumber\\
&~+ 
\pi_{1}( s_{0}\sin^{2}\varphi_{1}+s_{0}- s_{0}^{2}\cos^{2}\varphi_{1}). 
\end{align}
\begin{theorem}
Consider the quantum classification problem with training set 
$\rho^{\otimes \pi_{0}n}\otimes \sigma^{\otimes \pi_{1} n}$ where 
$\rho,\sigma$ are unknown qubit states and $(\pi_{0},\pi_{1})$ are known. 

Under the assumption \eqref{eq.assumption}, the asymptotic renormalised maximum risk $R_{max} (\tilde{\vec{z}}~; ~\rho_{0},\sigma_{0})$ of the plug-in classifier \eqref{eq.plug.in.estimation} is given by \eqref{eq.risk.plugin.estimation}.
\end{theorem}

Comparing the minimax risk \eqref{eq.optimal.risk} with the risk \eqref{eq.risk.plugin.estimation} of the plug-in classifier we get
$$
R_{max} (\tilde{\vec{z}}~; ~\rho_{0},\sigma_{0})-
R^{(l)}_{minmax}(\rho_{0},\sigma_{0})= 
\pi_{0}r_{0}(1\pm \sin\varphi_{0})^{2} + 
\pi_{1}s_{0}(1\mp \sin\varphi_{1})^{2},
$$
with the signs are chosen according to the sign of 
$\pi_{0}r_{0}\sin\varphi_{0}-\pi_{1}s_{0}\sin\varphi_{1}$. This quantity is equal to zero if and only if $\sin\varphi_{0}=\mp 1$ and $\sin\varphi_{1}=\pm 1$ which means that the vectors $\vec{r}_{0} $ and $\vec{s}_{0}$ are parallel and point in the same  direction. For fixed priors, the difference is maximum when the $\vec{r}_{0}$ and $\vec{s}_{0}$ point in opposite directions and have length one. 

This can be easily understood from the Gaussian model. When the vectors are parallel then learning requires an optimal joint measurement of {\it non-commuting} variables $(Q_{1}-Q_{2}, P_{1}-P_{2})$ whose risk is the same as that of heterodyning the oscillators first and constructing linear combinations. In the anti-parallel case we  need to measure {\it commuting} variables $(Q_{1}+Q_{2}, P_{1}-P_{2})$ which can be done directly, without any loss.

\subsection{The case of unknown priors}\label{sec.priors}
The analysis so far deals with known priors $\pi_{0},\pi_{1}$, which is the 
standard set-up usually considered in quantum statistics. In general, the priors may be unknown but can be estimated from the training set with a standard $n^{-1/2}$ error. Since the Helstrom measurement depends also on $(\pi_{0},\pi_{1})$, this uncertainty will bring an additional contribution to the excess risk. To find it, one needs to go back to the derivation of the quadratic loss function and add another unknown local parameter $\delta$ for the prior: $\pi_{0}= q_{0}+\delta/\sqrt{n}$.

Then 
\eqref{eq.oracle.bloch} becomes
\begin{equation}
\vec{p} 
= \frac{\vec{d} }{\| \vec{d} \|} 
= 
\frac{
\vec{d}_{0} + 
\frac{\vec{z}}{\sqrt{n}} +
\frac{\delta (\vec{r}_{0}+\vec{s}_{0} )}{\sqrt{n}}
}
{\left\|
\vec{d}_{0} + \frac{\vec{z}}{\sqrt{n}} +
\frac{\delta (\vec{r}_{0}+\vec{s}_{0} )}{\sqrt{n}}.
\right\|},
\end{equation}
By going through the same steps, we get to the quadratic loss function

\begin{equation}\label{loss.fct.quadratic.general}
L((\vec{u},\vec{v},\delta), \hat{\vec{z}}_n):= 
\frac{1}{4 \|\vec{d}_{0} \|} \| \vec{z}_{\perp} +\delta (\vec{r}_{0}+\vec{s}_{0})_{\perp}-\hat{\vec{z}}_n\|^{2}, 
\end{equation}
where $(\vec{r}_{0}+\vec{s}_{0})_{\perp}$ is the component orthogonal to 
$\vec{p}_{0}$. 

As before, the training set can be cast into a Gaussian model, with an additional independent component $Z\sim N(\delta,\pi_{0}\pi_{1})$. This means that when taking the expectation of $L$ we get an additional factor
$$
\frac{\|(\vec{r}_{0}+\vec{s}_{0})_{\perp} \|^{2}}{4 \|\vec{d}_{0} \|} Var(Z)= 
\frac{\pi_{0}\pi_{1} \|(\vec{r}_{0}+\vec{s}_{0})_{\perp} \|^{2}}{4 \|\vec{d}_{0} \|}. 
$$

\section{Conclusions}

We solved the problem of classifying two qubit states in the asymptotic local minimax statistical framework. Asymptotically the problem reduces to that of optimally estimating a sub-parameter of a quantum Gaussian model consisting of two independent oscillators in displaced thermal states with unknown means. The estimator is then used to construct an approximation of the (unknown) 
Helstrom measurement which is used to classify unlabelled states. The optimal procedure has excess risk of order $n^{-1}$ and we computed the exact constant factor $R^{(l)}_{minmax}(\rho_{0},\sigma_{0})$ as function of the two unknown states. Except in the special case of states with parallel Bloch vectors, the optimal procedure performs strictly better than the plug-in classifier obtained by estimating the states and applying the corresponding Helstrom measurement. The difference is only a constant factor, but it would probably become significant in more interesting infinite dimensional models.

Finally let us briefly discuss the Bayesian analogue of our result. In the Bayesian framework one would choose a `regular' prior $\mu(d\rho \times d\sigma)$  over 
the two types of states and try to find the (asymptotically) optimal Bayes risk for this prior
$$
R^{\mu}_{opt} := \limsup_{n\to\infty} \inf_{\widehat{M}_{n}} \,n \,
\int \mu(d\rho \times d\sigma) R_{(\rho,\sigma)}(\widehat{M}_{n}).
$$ 
When the states are pure and the prior is uniform, this has been done (even non-asymptotically) in \cite{Hayashi&Horibe}, but the proof relies on the symmetry of the prior and cannot be applied to general ones, and mixed states. Based on a similar analysis done for state estimation\cite{Belavkin&Guta}, we expect that our result can be used to prove that 
$$
R^{\mu}_{opt}=\int R^{(l)}_{minmax}(\rho_{0},\sigma_{0})\mu(d\rho_{0} \times d\sigma_{0}).
$$
The intuitive explanation is that when $n\to\infty$ the features of the prior $\mu$ 
are washed out and the posterior distribution concentrates in a local neighbourhood of the true parameter, where the behaviour of the classifiers is governed by the local minimax risk. Proving this relation is however beyond the scope of this paper. 

%
\ack

We thank Richard Gill for useful discussions. M.G. was supported by the EPSRC Fellowship EP/E052290/1. 

\section*{References}

\begin{thebibliography}{10}
\expandafter\ifx\csname url\endcsname\relax
  \def\url#1{{\tt #1}}\fi
\expandafter\ifx\csname urlprefix\endcsname\relax\def\urlprefix{URL }\fi
\providecommand{\eprint}[2][]{\url{#2}}

\bibitem{Mitchell97}
Mitchell T 1997 {\em Machine Learning\/} 1st ed (McGraw-Hill Education)

\bibitem{DevroyeGyorfiLugosi96}
Devroye L, Gy\"orfi L and Lugosi G 1996 {\em A Probabilistic Theory of Pattern
  Recognition\/} 1st ed (Springer)

\bibitem{Vapnik98}
Vapnik V 1998 {\em Statistical Learning Theory\/} (Wiley)

\bibitem{FriedmanHastieTibshirani03}
Friedman J~H, Hastie T and Tibshirani R 2003 {\em Elements of Statistical
  Learning: Data Mining, Inference, and Prediction\/} (Springer)

\bibitem{Bishop}
Bishop C~M 2006 {\em Pattern recognition and machine learning\/} (Springer)

\bibitem{NC00}
Nielsen M and Chuang I 2000 {\em Quantum Computation and Quantum Information\/}
  (Cambridge: Cambridge University Press)

\bibitem{Wiseman&Milburn}
Wiseman H~M and J M~G 2009 {\em Quantum measurements and control\/} (Cambridge
  University Press)

\bibitem{Smithey}
Smithey D~T, {Beck, M}, {Raymer, M~G} and {Faridani, A} 1993 {\em Phys. Rev.
  Lett.\/} {\bf 70} 1244--1247

\bibitem{Haffner}
H\"{a}ffner H, H\"{a}nsel W, Roos C~F, Benhelm J, Chek-al kar D, Chwalla M,
  K\"{o}rber T, Rapol U~D, Riebe M, Schmidt P~O, Becher C, G\"{u}hne O, D\"{u}r
  W and Blatt R 2005 {\em Nature\/} {\bf 438} 643--646

\bibitem{Holevo}
Holevo A~S 1982 {\em Probabilistic and Statistical Aspects of Quantum Theory\/}
  (North-Holland)

\bibitem{Helstrom}
Helstrom C~W 1976 {\em Quantum Detection and Estimation Theory\/} (Academic
  Press, New York)

\bibitem{Leonhardt}
Leonhardt U 1997 {\em Measuring the Quantum State of Light\/} (Cambridge
  University Press)

\bibitem{Hayashi.editor}
Hayashi M (ed) 2005 {\em Asymptotic theory of quantum statistical inference:
  selected papers\/} (World Scientific)

\bibitem{Paris.editor}
Paris M~G~A and {\v R}eh{\'a}{\v c}ek J (eds) 2004 {\em {Quantum State
  Estimation}\/}

\bibitem{Barndorff-Nielsen&Gill&Jupp}
Barndorff-Nielsen O~E, {Gill, R} and {Jupp, P~E} 2003 {\em J. R. Statist. Soc.
  B\/} {\bf 65} 775--816

\bibitem{SasakiCarlini2002}
Sasaki M and Carlini A 2002 {\em Physical Review A\/} {\bf 66} 022303

\bibitem{BergouHillery2005}
Bergou J and Hillery M 2005 {\em Phys. Rev. Lett.\/} {\bf 94} 160501

\bibitem{Hayashi&Horibe}
Hayashi A, Horibe M and Hashimoto T 2005 {\em Phys. Rev. A\/} {\bf 72} 052306

\bibitem{Hayashi_et_al2006}
A~Hayashi M~Horibe T~H 2006 {\em Phys. Rev. A\/} {\bf 93} 012328

\bibitem{AimeurBrassardGambs2006}
A{\"i}meur E, Brassard G and Gambs S 2006 {\em Proc. of the 19th Canadian
  Conference on Artificial Intelligence (Canadian AI'06)\/} (Qu{\'e}bec City,
  Canada: Springer) pp 433--444

\bibitem{AimeurBrassardGambs2007}
A{\"i}meur E, Brassard G and Gambs S 2007 {\em Proc. of the 24th
  {I}nternational {C}onference of {M}achine {L}earning ({ICML}'07)\/}
  (Corvallis, USA) pp 1--8

\bibitem{Gambs2008}
Gambs S 2008 Quantum classification arXiv:0809.0444

\bibitem{Guta&Kahn}
Gu\c{t}\u{a} M and Kahn J 2006 {\em Phys. Rev. A\/} {\bf 73} 052108

\bibitem{Guta&Janssens&Kahn}
Gu\c{t}\u{a} M, Janssens B and Kahn J 2008 {\em Commun. Math. Phys.\/} {\bf
  277} 127--160

\bibitem{Guta&Kahn2}
Kahn J and Gu\c{t}\u{a} M 2009 {\em Commun. Math. Phys.\/} {\bf 289} 597--652

\bibitem{Guta&Jencova}
Gu\c{t}\u{a} M and Jen\c{c}ov\'{a} A 2007 {\em Commun. Math. Phys.\/} {\bf 276}
  341--379

\bibitem{LeCam}
Le~Cam L 1986 {\em Asymptotic Methods in Statistical Decision Theory\/}
  (Springer Verlag, New York)

\bibitem{Guta&Adesso}
Gu\c{t}\u{a} M, Adesso G and Bowles P {Quantum teleportation benchmarks for
  independent and identically-distributed spin states and displaced thermal
  states} in preparation

\bibitem{Bagan_et_al2004}
Bagan E, Baig M, Mu\~noz Tapia R and Rodriguez A 2004 {\em Phys. Rev. A\/} {\bf
  69} 010304

\bibitem{Bagan_et_al2006}
Bagan E, Ballester M~A, Gill R~D, Mu\~noz Tapia R and Romero-Isart O 2006 {\em
  Phys. Rev. Lett.\/} {\bf 97} 130501

\bibitem{Bagan&Gill}
Bagan E, {Ballester, M A}, {Gill, R D}, {Monras, A} and {Mun\~ oz-Tapia, R}
  2006 {\em Phys. Rev. A\/} {\bf 73} 032301

\bibitem{Gammelmark&Molmer}
Gammelmark S and Molmer K 2009 {\em New Journal of Physics\/} {\bf 11} 033017

\bibitem{Bisio&Chiribella}
Bisio A, Chiribella G, D'Ariano G~M, Facchini S and Perinotti P 2010 {\em Phys.
  Rev. A\/} {\bf 81} 032324

\bibitem{AudibertTsybakov07}
Audibert J~Y and Tsybakov A~B 2007 {\em Annals of Statistics\/} {\bf 35}
  608--633

\bibitem{Cohn}
Cohn D~A, Ghahramani, Z and Jordan M~I 1996 {\em Journal of Artificial
  Intelligence Research\/} {\bf 4} 129--145

\bibitem{Belavkin}
Belavkin V~P 1976 {\em Theor. Math. Phys.\/} {\bf 26} 213--222

\bibitem{vanderVaart}
van~der Vaart A 1998 {\em Asymptotic Statistics\/} (Cambridge University Press)

\bibitem{Radcliffe}
Radcliffe J~M 1971 {\em J. Phys. A\/} {\bf 4} 313--323

\bibitem{Kitagawa&Ueda}
Kitagawa M and Ueda M 1993 {\em Phys. Rev. A\/} {\bf 47} 5138--5143

\bibitem{Guta&Kahn3}
Gu\c{t}\u{a} M and Kahn J {Optimal state estimation: attainability of the
  Holevo bound} in preparation

\bibitem{DAriano_et_al2005}
D'Ariano G~M, Presti P~L and Perinotti P 2005 {\em Journal of Physics A\/} {\bf
  38} 5979

\bibitem{Belavkin&Guta}
Gill R~D 2008 {\em Quantum Stochastics and Information: Statistics, Filtering
  and Control\/} ed Belavkin V~P and Guta M (World Scientific, Singapore) pp
  239--261

\end{thebibliography}
\providecommand{\newblock}{}

\end{document}